\documentclass[final]{siamltex}

\usepackage{amsmath}
\usepackage{a4wide}
\usepackage{amsfonts}
\usepackage{graphicx}
\usepackage{subfig}
\usepackage{array}
\usepackage{placeins}

\newcommand{\scalar}[2]{\langle#1\,,#2\rangle}

\newcommand{\norm}[1]{\|#1\|}

\newcommand{\vol}{\mathrm{vol}}
\renewcommand{\H}{\mathcal{H}}
\newcommand{\D}{\mathcal{D}}
\newcommand{\C}{\mathcal{C}}

\newcommand{\pM}{\partial M}

\title{\sc Numerical Solutions of the spectral problem for arbitrary self-adjoint extensions of the 1D Schr\"odinger equation\thanks{This work has been partially supported by the Spanish MICINN grant MTM2010-21186-C02-02, QUITEMAD programme P2009 ESP-1594.  AI has been partially supported  by Fundaci\'on Caja Madrid. JMP has been supported by the UCIIIM University through the PhD Program Grant M02-0910}}

\author{A. Ibort \footnotemark[2]
\and J. M. P\'erez-Pardo \footnotemark[3]}
\date{}

\begin{document}
\maketitle

\renewcommand{\thefootnote}{\fnsymbol{footnote}}

\footnotetext[2]{Department of Mathematics, Univ. of California at Berkeley, Berkeley CA 94720, USA.   On leave of absence from Depto. de Matem\'aticas, Univ. Carlos III de Madrid, Avda. de la Universidad 30, 28911 Legan\'es, Madrid, Spain.}
\footnotetext[3]{Depto. de Matem\'aticas, Univ. Carlos III de Madrid, Avda. de la Universidad 30, 28911 Legan\'es, Madrid, Spain.}

\renewcommand{\thefootnote}{\arabic{footnote}}

\begin{abstract}  A numerical algorithm to solve the spectral problem for arbitrary self-adjoint extensions of 1D regular Schr\"odinger operators is presented.  It is shown that the set of all self-adjoint extensions of 1D regular Schr\"odinger operators is in one-to-one correspondence with the group of unitary operators on the finite dimensional Hilbert space of boundary data, and they are characterized by a generalized class of boundary conditions that include the well-known Dirichlet, Neumann, Robin, (quasi-)periodic boundary conditions, etc.    The numerical algorithm is based on a nonlocal boundary extension of the finite element method and their convergence is proved.   An appropriate basis of boundary functions must be introduced to deal with arbitrary boundary conditions and the conditioning of its computation is analyzed.    Some significant numerical experiments are also discussed as well as the comparison with some standard algorithms.   In particular it is shown that appropriate perturbations of standard boundary conditions for the free particle leads to the theoretically predicted result of very large absolute values of the negative groundlevels of the system as well as the localization of the corresponding eigenvectors at the boundary (edge states).
\end{abstract}

\pagestyle{myheadings}
\thispagestyle{plain}
\markboth{A. Ibort and J. M. P\'erez-Pardo}{Self-adjoint extensions of 1D Schr\"odinger operators}

\section{Introduction}

The study of the self-adjointness of Schr\"odinger operators has been a fundamental mathematical problem since the beginning of Quantum Mechanics and there is a vast literature on the subject (see for instance \cite{Re75}, the review \cite{Si00} and references therein).   In spite of this, there is a continuous flow of new results and even surprises (see for instance the recent papers where some apparent paradoxical aspects of the spectrum of certain self-adjoint extensions of the Schr\"odinger operator in 2D are analyzed \cite{Be08}, \cite{Ma09}, \cite{Be09}).   

Consider the evolution of a quantum system on a $D$-dimensional Riemannian manifold $M$ with boundary $\partial M$ under the influence of a potential $\mathcal{V}$ which is given by the Schr\"odinger equation
$i\hbar \frac{\partial \Psi}{\partial t} = H \Psi ,$ with $H$ the Hamiltonian operator of the system given by 
\begin{equation}\label{hamiltonian} H = -\frac{\hbar^2}{2 m} \Delta_\eta + \mathcal{V}(x) =-\frac{\hbar^2}{2 m} \frac{1}{\sqrt{| \eta |}} \frac{\partial }{\partial x^j} \sqrt{| \eta | }\eta^{jk} \frac{\partial }{\partial x^k} + \mathcal{V}(x).
\end{equation}
Here $\Delta_\eta$ stands for the Laplace-Beltrami operator on $M$ defined by the metric tensor $\eta$ given by $\eta = \eta_{jk} (x) dx^k dx^j$, $| \eta | = | \det{ \eta_{jk}(x)} | $ and $\eta^{ij} \eta_{jk} = \delta_k^i$.
The second order differential operator $H$ is formally self-adjoint, however in order to define a unitary evolution of the quantum state $\Psi$, a self-adjoint extension of it must be specified.    If $H$ denotes one of such self-adjoint extensions, because of Stone's theorem, a one-parameter group $U_t$ of unitary operators exists such that $U_t = \exp (-it H/\hbar)$ and the quantum evolution of a given initial state $\Psi_0$ will be uniquely determined as $\Psi_t = U_t\Psi_0$.     Self-adjoint extensions of the operator $H$ are usually chosen by fixing the boundary values of the functions where the operator $H$ acts, typically  Dirichlet or Neumann boundary conditions.  However they are by no means the most general choice of boundary conditions determining self-adjoint extensions of the operator $H$ and a variety of other possibilities exist.    The development of quantum information technologies makes relevant the discussion of more general classes of boundary conditions, i.e., of general self-adjoint extensions for the operator $H$, as well as the numerical computation of their spectrum in order to integrate Schr\"odinger's equation.  For instance it has been recently proposed a physical implementation of a universal quantum computer by using scattering states on a quantum graph \cite{Ci09}, i.e., of a Schr\"odinger operator on a graph.  It is also relevant to mention here the celebrated exponential algorithmic speedup by a quantum walk by Childs {\em et al} \cite{Ci03}.   The self-adjoint extensions defining a quantum system on a graph, otherwise called quantum Kirchoff's rules \cite{Ko00}, \cite{Ko03} are just one more instance of these possibilities.  Another one is provided by the use of absorbing boundary conditions \cite{Ba04} or by diverse optical cavities implementations  \cite{Kn03} of quantum walks.  \\

In this paper we present an algorithm based on the finite element method to solve numerically the spectral problem for {\em all\,}  possible self-adjoint extensions of 1D regular Schr\"odinger operators.  In this particular instance, the Schr\"odinger operator takes the form of a Sturm-Liouville operator. There exist many methods for approaching this problem (see \cite{Ac09}, \cite{Ch99}, \cite{Ch09}, \cite{Fa57}, \cite{Pr73} just to mention a few of them), however they cannot be used to solve the problem of general self-adjoint extensions of it, as explained below.  Most of the previous algorithms solve the wide class of self-adjoint extensions given by the equations of the form:
$\alpha\cdot \Psi(a)+\beta\cdot \Psi'(a)=0$, $\gamma\cdot \Psi(b)+\delta\cdot \Psi'(b)=0$,
where $a$ and $b$ are the endpoints of the interval.   For instance, (quasi)-periodic boundary conditions can not be implemented with the above setting. 

In this work we will take the approach developed in \cite{As05} to describe the set of self-adjoint extensions of Laplace-Beltrami operators.   We will show in section \ref{Selfadjoint_extensions_of_Schrodinger_operators}  that in 1D the set of self-adjoint extensions of the Schr\"odinger operator $H$ is in one-to-one correspondence with the group of unitary operators on the space of boundary data for the problem and we will provide explicit constructions of such correspondence so that it can be actually used in the implementation of the corresponding numerical algorithms.  Similar results are valid in higher dimensions, however there are additional technical difficulties related to the functional spaces where the domains of the operators are defined, that make the presentation and the discussion of the results more difficult. Because of this and in order to make the presentation of the results as simple and as short as possible we will consider in this paper just the 1D situation leaving the discussion in higher dimensions for later works. Thus the manifold $M$ will be assumed to be one-dimensional and compact and the potential function $\mathcal{V}$ will be assumed to be regular and bounded below.

The numerical algorithm which is developed to solve the eigenvalue problem for all self-adjoint extensions of the Hamiltonian operator $H$ described in terms of boundary data is based on the finite element method, even though the treatment of the boundary basis functions is novel.   The determination of the boundary basis function is analyzed as well as its conditioning.   The stability and convergence of the method is proved by describing the corresponding self-adjoint extensions in weak form.   A number of experiments are presented that are consistent with the obtained results and shows the efficiency of the corresponding computer algorithms.   These results are compared with the standard methods in some significant examples showing again the good behavior of the proposed algorithm.   In particular it is shown that for a family of self-adjoint extensions the groundlevel of the system must go to $-\infty$, a behaviour that is described nicely by the algorithm presented here as well as the presence of {\em edge states}, i.e., eigenfuctions which are localized at the boundary.  

The paper is organized as follows. Section \ref{Selfadjoint_extensions_of_Schrodinger_operators} is devoted to prove the general theorems that will allow for a convenient description of the self-adjoint extensions of the Schr\"odinger operator $H$ in 1D in terms of boundary data.  The main results of the finite element method for the eigenvalue problem for completely general self-adjoint extensions are discussed in section \ref{FEM}. An appropriate basis of nonlocal boundary functions is introduced there in order to implement the boundary conditions as described by the general theory. Moreover, in this section, results concerning the convergence and the stability of the numerical scheme are also proved. In section \ref{Numerical_experiments_and_conclusions} we present some numerical experiments displaying relevant features of the algorithm, as well as a comparison with other methods.

%%%%%%%%%%%%%%%%%%%%%%%%%%%%%%%%%%%%%%%%%%%%%%%%%%%%%%%%%%%%%%%%%%%%%%%%%%%%%%%%%%%%%%%%%%%%%%%%%%%%%%%%%%%%%%

\section{Self-adjoint extensions of Schr\"odinger operators}\label{Selfadjoint_extensions_of_Schrodinger_operators}

As it was stated in the introduction, we will restrict our attention to the case of Schr\"odinger operators on 1D compact manifolds and regular potentials bounded below.  The Schr\"odinger operator for a particle moving on a smooth manifold $M$ with boundary $\partial M$ and Riemannian metric $\eta$ is given by the Hamiltonian operator $H$ defined in Eq.\ (\ref{hamiltonian}). 
Notice first that a compact 1D manifold $M$ consists of a finite number of closed intervals $I_\alpha$, $\alpha = 1,\ldots,n$.  Each interval will have the form $I_\alpha = [a_\alpha, b_\alpha] \subset \mathbb{R}$ and the boundary of the manifold $M = \coprod_{\alpha=1}^n [a_\alpha, b_\alpha]$ is given by the family of points $\partial M = \{ a_1, b_1, \ldots, a_n,b_n\}$.     Functions $\Psi$ on $M$ are determined by vectors $(\Psi_1, \ldots, \Psi_n)$ of complex valued functions $\Psi_\alpha \colon I_\alpha \to \mathbb{C}$.
A Riemannian metric $\eta$ on $M$ is given by specifying a Riemannian metric $\eta_\alpha$ on each interval $I_\alpha$, this is, by a positive smooth function $\eta_\alpha(x) > 0$ on the interval $I_\alpha$, i.e., $\eta|_{I_\alpha} = \eta_\alpha (x) dx^2$.    Then the $L^2$ inner product on $I_\alpha$ takes the form $\langle \Psi_\alpha , \Phi_\alpha \rangle = \int_{a_\alpha}^{b_\alpha} \overline{\Psi}_\alpha (x) \Phi_\alpha (x) \sqrt{\eta_\alpha (x)} dx$ and the Hilbert space of square integrable functions on $M$ is given by $L^2(M) = \bigoplus_{\alpha= 1}^n L^2(I_\alpha, \eta_\alpha)$.

On each subinterval $I_\alpha = [a_\alpha, b_\alpha]$ the differential operator $H_\alpha = H|_{I_\alpha}$ takes the form 
of a Sturm-Liouville operator 
\begin{equation}\label{sturm}
H_\alpha = - \frac{1}{W_\alpha} \frac{d}{dx} p_\alpha(x) \frac{d}{dx} + V_\alpha(x),
\end{equation}
with smooth coefficients $W_\alpha = 1/(2\sqrt{\eta_\alpha}) > 0$ (in what follows we are taking the physical constants $\hbar$ and $m$ such that $\hbar^2/2m = 1$), $p_\alpha(x) = 1/\sqrt{\eta_\alpha}$, and is formally self-adjoint in the sense that
\begin{equation}\label{laplace}
\int_M\overline{\Psi}\Delta_\eta\Phi\sqrt{\eta} dx=-\int_M \frac{d\overline{\Psi}}{dx}\frac{d\Phi}{dx}\sqrt{\eta}dx=\int_M (\overline{\Delta_\eta\Psi})\Phi\sqrt{\eta} dx,
\end{equation}
for $\Psi$, $\Phi\in\C_c^\infty(M\backslash \pM)$ complex-valued smooth functions with compact support in the interior of $M$. In fact, the differential expression \eqref{laplace} defines a symmetric operator on $L^2(M)$ with dense domain $\C_c^\infty(M\backslash\pM)$.   The operator $\Delta_\eta$ is closable on $L^2(M)$ and its minimal closed extension, that we will keep denoting by $\Delta_\eta$, has domain $D_0 := \mathcal{H}_0^2(M)$, the Sobolev space of order 2 on $M$ with functions and first derivatives vanishing at the boundary.  The adjoint operator  $\Delta_\eta^\dagger$ has dense domain $D_0^\dagger = \{ \Phi\in L^2(M) \mid \Delta_\eta\Phi  \in L^2(M) \}$ on $L^2(M)$.   In the particular instance of 1D the domain $D_0^\dagger$ of the adjoint operator $\Delta_\eta^\dagger$ is exactly given by $\mathcal{H}^2(M)$, the Sobolev space of order 2 on $M$, i.e., the space of functions on $L^2(M)$ possessing first and second weak derivatives that are square integrable (notice that in higher dimensions this is not true because the domain of the adjoint could contain functions which are not in $\mathcal{H}^2(M)$, see for instance \cite{Gr68}).   Clearly $D_0 = \mathcal{H}_0^2(M)  \subset  \mathcal{H}^2 (M) = D_0^\dagger$ and $\Delta_\eta \subset \Delta_\eta^\dagger$, thus the operator $\Delta_\eta$ is symmetric but not self-adjoint on $L^2(M)$.   Notice that the maximal extension of the operator $\Delta_\eta$, this is the operator $\Delta_\eta^\dagger$ defined in $D_0^\dagger$ is not self-adjoint either because $(D_0^\dagger)^\dagger = D_0 \subset D_0^\dagger$.

Hence we would like to determine whether there exist self-adjoint extensions of the symmetric operator $\Delta_\eta$, i.e., operators $\Delta_D$ with domain $D$ such that $D_0 \subset D \subset D_0^\dagger$, $\Delta_D\mid _{D_0} = \Delta_\eta$ and $\Delta_D = \Delta_D^\dagger$.
As it will become clear from the discussion below, there exist self-adjoint extensions of $\Delta_\eta$ and all of them are determined by appropriate boundary conditions for the functions in $D_0^\dagger = \mathcal{H}^2(M)$. 

Von Neumann's theorem establishes (see for instance \cite{We80}, Thm. 8.12) that there is a one-to-one correspondence between self-adjoint extensions $\Delta_D$ of the Laplace-Beltrami operator $\Delta_\eta$ and unitary operators $K \colon \mathcal{N}_+ \to \mathcal{N}_-$, where the deficiency spaces $\mathcal{N}_\pm$ are defined as:
\begin{equation}\label{deficiency}
 \mathcal{N}_\pm = \{ \Psi \in \mathcal{H}^2(M)  \mid  \Delta_\eta^\dagger \Psi = \pm i \Psi \} .
 \end{equation}
Thus, given a unitary operator $K\colon \mathcal{N}_+ \to \mathcal{N}_-$, the domain $D$ of the operator $\Delta_D$ is given by 
\begin{equation}\label{Neumanndomain}
D = D_0 \oplus (I + K)\mathcal{N}_+,\end{equation} 
and the extended operator $\Delta_D$ takes the explicit form:
\begin{equation}\label{neumann_op}
 \Delta_D (\Psi_0 \oplus (I+ K)\xi_+ ) = \Delta_\eta\Psi_0 \oplus i(I-K) \xi_+ ,
 \end{equation}
for all $\Psi_0 \in D_0$ and $\xi_+\in \mathcal{N}_+$.
 
Unfortunately von Neumann's theorem is not always suitable for the numerical computation of general self-adjoint extensions of the Laplace-Beltrami operator because one needs to determine first the deficiency spaces $\mathcal{N}_\pm$ and, moreover, the domain given by Eq.\ \eqref{Neumanndomain} is algebraically not well behaved.   We can take however a different route inspired in the classical treatment of formally self-adjoint differential operators.   If we rewrite Eq.\ (\ref{laplace}) for functions $\Psi$, $\Phi$ $\in D^\dagger = \mathcal{H}^2(M)$ instead of $C_c^\infty(M \backslash \partial M)$, a simple computation shows:

\begin{equation}\label{boundary}
 \int_M \overline{\Psi} \, \Delta_\eta \Phi \, \sqrt{\eta}dx = \int_M (\overline{\Delta_\eta \Psi})\,  \Phi \, \sqrt{\eta}dx + \int_{\partial M} \left( \overline{\psi} \dot{\varphi} - \overline{\dot{\psi}} \varphi \right) \vol_{\partial \eta},
\end{equation}
where $\vol_{\partial \eta}$ stands for the induced Riemannian metric at the boundary, $\varphi := \Phi\mid_{\partial M}$, $\dot{\varphi} := \frac{d \Phi}{d \nu}\mid_{\partial M}$, where $\nu$ is the exterior normal vector to $\pM$, and respectively for $\psi$ and $\dot{\psi}$. 

We thus obtain the Lagrange boundary form $\Sigma$ for the Laplace-Beltrami operator:
\begin{equation}\label{lagrange}
 \Sigma ((\psi,\dot{\psi}), (\varphi,\dot{\varphi})) = \int_{\partial M} \overline{\psi} \dot{\varphi}\;\vol_{\partial \eta} -\int_{\pM} \overline{\dot{\psi}} \varphi\; \vol_{\partial \eta} = \langle \psi, \dot\varphi \rangle_{L^2(\partial M)} - \langle \dot{\psi}, \varphi \rangle_{L^2(\partial M)} .
 \end{equation}
In what follows, if there is no risk of confusion, we will omit the subscript $L^2(\partial M)$ that denotes the $L^2$ inner product on the boundary manifold $\partial M$ with respect to the measure defined by the volume form $\vol_{\partial \eta}$ and we will simply write $\langle \psi, \varphi \rangle = \int_{\partial M} \overline{\psi} \varphi \,  \vol_{\partial \eta}$.    The Lagrange boundary bilinear form $\Sigma$ defines a continuous bilinear form on the Hilbert space $L^2(\partial M)\oplus L^2(\partial M)$,
\begin{equation}\label{lagrange2}
 \Sigma ((\psi_1,\psi_2), (\varphi_1,\varphi_2)) = \langle \psi_1, \varphi_2 \rangle - \langle \psi_2, \varphi_1 \rangle, \quad \quad \forall (\psi_1,\psi_2), (\varphi_1,\varphi_2) \in  L^2(\partial M)\oplus L^2(\partial M).
 \end{equation}
The trace map $\gamma \colon \mathcal{H}^2 (M) \to L^2(\partial M)\oplus L^2(\partial M)$ given by $\gamma (\Psi ) = (\psi, \dot{\psi})$, is continuous and induces an isomorphism  from $\mathcal{H}^2(M) / \ker \gamma$ onto $\H^{3/2}(\partial M) \oplus \H^{1/2}(\partial M) \subset L^2(\partial M)\oplus L^2(\partial M)$ (see for instance \cite{Ad75}, Thm. 7.20), where $\H^{3/2}$ and $\H^{1/2}$ are Sobolev spaces of fractional order.  The previous observations provide a simple characterization of self-adjoint extensions of the operator $\Delta_\eta$ in 1D.\\

\begin{theorem}\label{self1}
There is a one-to-one correspondence between self-adjoint extensions $\Delta_D$ of the Laplace-Beltrami operator $\Delta_\eta$ and non-trivial maximal closed isotropic subspaces $W$ of $\Sigma$ contained in $\H^{3/2}(\partial M) \oplus \H^{1/2}(\partial M)$.    The correspondence being explicitly given by $D \mapsto \gamma (D)$.
\end{theorem}\\

\begin{proof}  
Let $D\subset \mathcal{H}^2(M)$ be the domain of a self-adjoint extension $\Delta_D$ of the operator $\Delta_\eta$.  Consider the subspace $W := \gamma (D)  \subset \H^{3/2}(\partial M) \oplus \H^{1/2}(\partial M) \subset  L^2(\partial M)\oplus L^2(\partial M)$ consisting on the set of pairs of functions $(\varphi, \dot{\varphi})$ that are respectively the restriction to $\partial M$ of a function $\Phi\in D$ and its normal derivative.  Notice that the subspace $D$  is closed in $\mathcal{H}^2(M)$.  Hence, because $\gamma$ is an homeomorphism from $\mathcal{H}^2(M)/\ker \gamma$ to $\H^{3/2}(\partial M) \oplus \H^{1/2}(\partial M) $, $\gamma (D)$ is a closed subspace of $\H^{3/2}(\partial M) \oplus \H^{1/2}(\partial M) $.   Because of Eq.\ (\ref{boundary}), it is clear that $\Sigma\mid_{W} = 0$ and finally the subspace $W$ is maximally isotropic in $\H^{3/2}(\partial M) \oplus \H^{1/2}(\partial M)$, because if this were not the case, there will be a closed  isotropic subspace $W' \subset \H^{3/2}(\partial M) \oplus \H^{1/2}(\partial M)$  containing $W$.   Then $\gamma^{-1}(W')$ defines a domain $D'$ containing $D$ such that  the operator $\Delta_\eta\bigr|_{D'}$ would be symmetric on it, in contradiction with the self-adjointness assumption.

Let $W\subset L^2(\partial M)\oplus L^2(\partial M)$ be a maximal closed $\Sigma$-isotropic subspace in $\H^{3/2}(\partial M) \oplus \H^{1/2}(\partial M)$.  Then consider the subspace $D_W : =  \gamma^{-1}(W) \subset \mathcal{H}^2(M)$ of functions $\Psi$ such that $(\psi, \dot{\psi}) \in W$.   It is clear that for any pair of functions $\Psi$, $\Phi$ on $D_W$, because $W$ is isotropic with respect to $\Sigma$, Eq.\ (\ref{boundary}) gives $\langle \Psi, \Delta_\eta \Phi \rangle = \langle \Delta_\eta \Psi,\Phi \rangle$ and the operator $\Delta_\eta$ is therefore symmetric in $D_W$.  Moreover, because of the maximality of $W$ in $\H^{3/2}(\partial M) \oplus \H^{1/2}(\partial M)$  it is easy to see that $D_W = D_W^\dagger$, hence it is self-adjoint.   
\end{proof}\\
 
The key observation now is that, because $\operatorname{dim} M = 1$, $\H^{q}(\pM)=L^2(\pM)=\mathbb{C}^{2n}, \forall q$, where $n$ is the number of subintervals. Hence, the self-adjoint extensions of the Laplace-Beltrami operator are
in one-to-one correspondence with the maximal isotropic subspaces of the Lagrange boundary form $\Sigma$.  It is a well known result (see \cite{As05}, \cite{Br08} and references therein) that the maximal isotropic subspaces of the bilinear form \eqref{lagrange2} are in one-to-one correspondence with the graphs of unitary operators $U \colon L^2(\pM)\to L^2(\pM)$ through the relationship $$\varphi_1-i\varphi_2=U(\varphi_1+i\varphi_2).$$
This result is easily derived by realizing that the unitary transformation $C \colon   L^2(\partial M)\oplus L^2(\partial M) \to L^2(\partial M)\oplus L^2(\partial M)$, defined by $C(\varphi, \dot{\varphi}) = \bigl(\frac{1}{\sqrt{2}}(\varphi+i\dot{\varphi}), \frac{1}{\sqrt{2}}(\varphi-i\dot{\varphi})\bigr)$, transforms the Lagrange bilinear form  $\Sigma$ into the bilinear form:
\begin{equation}\label{graph_uni}
 \tilde{\Sigma}((\varphi_+,\varphi_-),(\psi_+,\psi_-))= -i \left[  \langle \varphi_+,\psi_+\rangle - \langle \varphi_-,\psi_-\rangle \right] ,
 \end{equation}
where the notation $\varphi_\pm = \frac{1}{\sqrt{2}}(\varphi\pm i \dot{\varphi})$ has been used.  Then it follows immediately that 
maximal isotropic closed subspaces of $\Sigma$ are mapped by $C$ into graphs of unitary operators of $L^2(\partial M)$. \\

The last step in discussing all self-adjoint extensions of the Hamiltonian $H$ consists in realizing that because $M$ is compact the operator multiplication by a regular function is essentially self-adjoint and its unique self-adjoint extension has domain $L^2(M)$.  Hence, the self-adjoint extensions of $H$ coincide with the self-adjoint extensions of $\Delta_\eta$.

We can summarize the preceding analysis by stating that, under the conditions above, the domain of a self-adjoint extension of the Schr\"odinger operator $H$ in 1D is defined as a closed subspace of functions $\Psi$ on $\mathcal{H}^2(M)$ satisfying:
\begin{equation}\label{asorey}
\psi - i\dot{\psi} = U (\psi + i \dot{\psi})
\end{equation}
for a given unitary operator $U \colon \mathbb{C}^{2n} \to \mathbb{C}^{2n}$.   Equation (\ref{asorey}) represent general quantum Kirchoff's rules, as they define the most general self-adjoint extension of the quantum free motion on a family of intervals.  
Notice for instance that $U = I$ corresponds to Neumann's boundary conditions and $U = - I$ determines Dirichlet's boundary conditions.  Other unitary matrices $U$ will define different graphs, but many other will not.  The formula above, Eq.\ (\ref{asorey}), provides a powerful and effective computational tool to deal with general self-adjoint extensions of Schr\"odinger operators, (even in dimensions higher than 1), and will be used extensively in the rest of the paper.

 Once we have determined a self-adjoint extension $H_U$ of the Schr\"odinger operator $H$, we can obtain its unitary evolution as indicated in the Introduction by computing the flow $U_t = \exp(-itH_U/\hbar )$.   It is well-known that the Dirichlet extension of the Laplace-Beltrami operator has a pure discrete spectrum because of the compactness of the manifold and the ellipticity of the operator, hence all self-adjoint extensions have a pure discrete spectrum (see \cite{We80}, Thm. 8.18).  Then the spectral theorem for the self-adjoint operator $H_U$ states:
$$ H_U = \sum_{k = 1}^\infty \lambda_k P_k ,$$
where $P_k$ is the orthogonal projector onto the finite-dimensional eigenvector space $V_k$ corresponding to the eigenvalue $\lambda_k$.  The unitary flow $U_t$ is given by:
$$ U_t = \sum_{k = 1}^\infty e^{-it\lambda_k / \hbar} P_k .$$
Hence all that remains to be done is to solve the eigenvalue problem:
\begin{equation}\label{eigenvalue_problem}
H_U \Psi = \lambda \Psi ,
\end{equation}
for the Schr\"odinger operator $H_U$.     We devote the rest of this paper to provide an efficient numerical algorithm to solve Eq.\ (\ref{eigenvalue_problem}).

%%%%%%%%%%%%%%%%%%%%%%%%%%%%%%%%%%%%%%%%%%%%%%%%%%%%%%%%%%%%%%%%%%%%%%%%%%%%%%%%%%%%%%%%%%%%%%%%%%%%%%%%%%%%%%%%%%%%%%%%%%

\section{Finite element method for the eigenvalue problem in the interval}\label{FEM}

Our aim now is to find a numerical algorithm to solve the spectral problem for the self-adjoint extension of the 1D Schr\"odinger operator $H$ determined by the unitary matrix $U \in U(2n)$ in (\ref{eigenvalue_problem}) and given explicitly  by the differential system:  
\begin{equation}\label{schrodinger}
\left\{	
	\begin{array}{l}
		- \Delta_\eta \Phi   + \mathcal{V}\Phi  =\lambda \Phi, \quad  \Phi\in \H^2(M)\,,     \\ \\ 
		\varphi-i\dot{\varphi}=U(\varphi+i\dot{\varphi}),\quad  \varphi=\Phi\bigr|_{\partial M}, \quad  \dot{\varphi}=\frac{d\Phi}{d\nu}\bigr|_{\partial M},\quad U\in U(2n). \\
	\end{array} \right.  .
\end{equation}

From now on we will keep the \emph{dot} notation for the outward normal derivative at the boundary $\dot{\varphi}=\frac{d\Phi}{d\nu }\bigr|_{\partial M}$ and use \emph{primes} to denote the standard derivative $\Phi'=\frac{d \Phi}{dx}$.  Because of the particularly simple form that the Lagrange boundary form takes in one dimension, it will be convenient to introduce the following notation (check with the right-hand side of  (\ref{lagrange})):
$$ [\overline{\Psi}\Phi']\bigr|_{\partial M} : = \langle \psi, \dot\varphi \rangle_{L^2(\partial M)} = \sum_{\alpha=1}^n \overline{\Psi}(b_\alpha)\Phi'(b_\alpha)-\overline{\Psi}(a_\alpha)\Phi'(a_\alpha) .$$
It is convenient to look for weak solutions of \eqref{schrodinger}.
Taking the inner product of \eqref{schrodinger} with a vector $\Psi$ on the dense domain $C^\infty (M)$ of smooth functions on $M$ and integrating by parts we obtain easily  the corresponding weak problem:
\begin{equation}\label{weak}
\left\{\begin{array}{l}	
	\langle \Psi' ,\Phi'\rangle-[\overline{\Psi}\Phi']_{\partial M}+\langle \Psi ,\mathcal{V}\Phi\rangle = \lambda\langle \Psi , \Phi \rangle \quad \forall \Psi\in C^\infty(M)\,, \\ \\
	\varphi-i\dot{\varphi}=U(\varphi+i\dot{\varphi}),\quad  \varphi=\Phi\bigr|_{\partial M}, \quad  \dot{\varphi}=\frac{d\Phi}{d\nu}\bigr|_{\partial M},\quad U\in U(2n). \\
	\end{array}\right.
\end{equation}
It is easy to show that $(\Phi, \lambda)$ is a solution of \eqref{schrodinger} iff it solves \eqref{weak}. In the next subsections we develop an algorithm that approximates such solutions and, moreover, we prove its stability and convergence.

\subsection{Finite elements for general self-adjoint boundary conditions}\label{finiteelementsforgeneralsabc}

To solve \eqref{weak} we construct a family of finite-dimensional subspaces of functions $S^N$ of $L^2(M)$ satisfying the boundary conditions \eqref{asorey}. Such finite-dimensional subspaces will be constructed using finite elements.   The finite element model $(K,\mathcal{P}, \mathcal{N})$ that we use is given by $K = [0,1]$ the unit interval, $\mathcal{P}$ the space of linear polynomials on $K$, and $\mathcal{N}$ the vertex set $\{ 0, 1\}$.     

The domain of our problem is the manifold $M$ which consists of the disjoint union of the intervals $I_\alpha = [a_\alpha, b_\alpha]$, $\alpha = 1, \ldots, n$.   For each $N$ we will construct a nondegenerate subdivision $M^N$ as follows. Let $r_\alpha$ be the integer defined as $r_\alpha = \left[ L_\alpha N /L \right] + 1$, where $[ x ]$ denotes the integer part of $x$, $L_\alpha = |b_\alpha - a_\alpha|$, and $L = L_1 + \ldots + L_n$.   We will assume that each $r_\alpha \geq 2$, and $N \geq 2n$.
  Let us denote by $r$ the multi index $(r_1, \ldots, r_n)$.  Then $|r| = r_1+ \cdots + r_n$ satisfies
$$ N \leq	|r | \leq N + n .$$
Now we will subdivide each interval $I_\alpha$ into $r_\alpha + 1$ subintervals of length: 
$$h_\alpha = \frac{L_\alpha} {r_\alpha + 1} .$$
The nondegeneracy condition supposes that it exists $\rho>0$ such that for all $I_{\alpha,a}\in M^N$ and for all $N>1$ $$\operatorname{diam}B_{I_{\alpha,a}}\geq \rho\operatorname{diam} I_{\alpha,a},$$ where $B_{I_{\alpha,a}}$ is the largest ball contained in $I_{\alpha,a}$ such that $I_{\alpha,a}$ is star shaped with respect to $B_{I_{\alpha,a}}$. In our particular case this is satisfied trivially since $$\operatorname{diam}B_{I_{\alpha,a}}=\operatorname{diam} I_{\alpha,a}\;.$$
(It would be possible to use a set of independent steps $h_\alpha$, one for each interval; however, this could create some
technical difficulties later on that we prevent in this way.)
Each subinterval $I_\alpha$ contains $r_\alpha + 2$ nodes that will be denoted as
$$x^{(\alpha)}_k = a_\alpha + h_\alpha k, \quad  k = 0, \ldots, r_\alpha +1 .$$ 

\subsubsection{Bulk functions}
Consider now the family $\mathcal{F}_r$ of  $|r|-2n$ piecewise linear functions $\{f_k^{(\alpha)}(x)\}_{k=2}^{r_\alpha- 1}$, $\alpha = 1, \ldots, n,\, 2 \leq k \leq r_\alpha - 1$, that are zero at all nodes except at the $k$th node of the interval $I_\alpha$, where it has value $1$, or more explicitly,
$$ f^{(\alpha)}_k (x) = \left\{   \begin{array}{ll}
s\,, & x = x^{(\alpha)}_{k-1} + s h_\alpha, \quad 0 \leq s \leq 1, \\
1-s\,, & x = x^{(\alpha)}_{k} + s h_\alpha, \quad 0 \leq s \leq 1, \\
0\,, & \mathrm{otherwise}\;.
\end{array}
\right. $$

Notice that these functions are differentiable on each subinterval.   All these functions satisfy trivially the boundary conditions \eqref{asorey} because they and their normal derivatives vanish at the endpoints of each interval.  They are localized around the inner nodes of the intervals.  We will call these functions \emph{bulk functions}.

\subsubsection{Boundary functions}\label{sectionboundaryfunctions}
We will add to the set of bulk functions a family of functions that implement nontrivially the boundary conditions determining the self-adjoint extension. These functions will be called \emph{boundary functions} and the collection of all of them will be denoted by $\mathcal{B}_r$.   
Contrary to bulk functions, boundary functions need to be ``delocalized'' so that they can fulfill any possible self-adjoint extension's boundary condition. 

%%%%%%%%%%%%%%%%%%%%%%%%%%%%%%%%%%%%%%%%%%%%%%%%%%%%%%%%%%%%%%%%%%%%

\begin{figure}[ht]
	\centering
	\includegraphics[height=4cm]{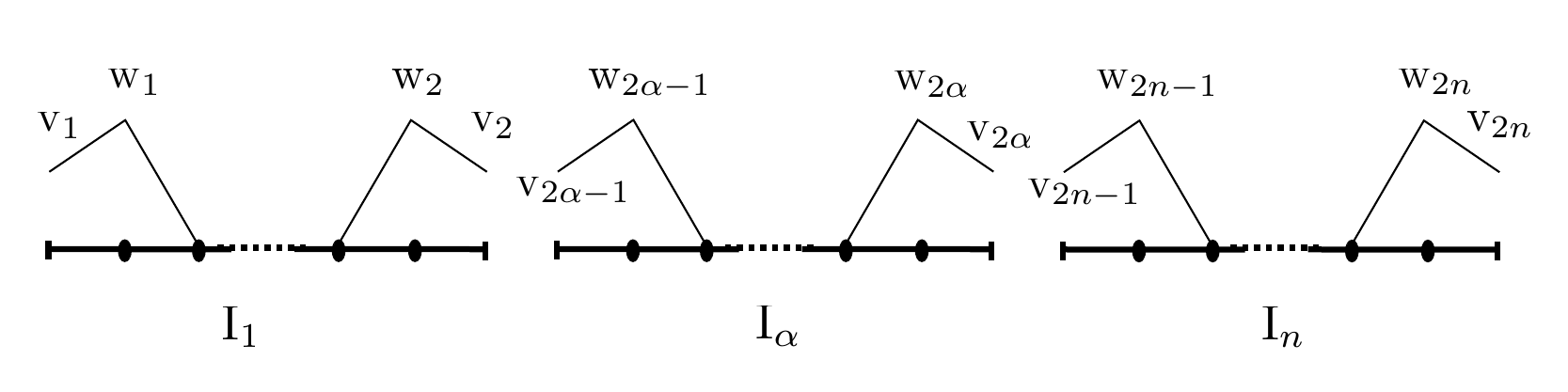}
	\caption{Boundary function $\beta^{(w)}$.}\label{function_beta}
\end{figure}

%%%%%%%%%%%%%%%%%%%%%%%%%%%%%%%%%%%%%%%%%%%%%%%%%%%%%%%%%%%%%%%%%%%%

Because the endpoints $x_0^{(\alpha)} = a_\alpha$, $x_{r_\alpha + 1}^{(\alpha)} = b_\alpha$ of the intervals $I_\alpha$ and the adjacent nodes, $x_1^{(\alpha)}$ and $x_{r_\alpha}^{(\alpha)}$, are going to play a prominent role in what follows, we introduce some notation that takes care of them.   We will consider an index $l = 1, \ldots, 2n$ that labels the endpoints of the intervals. Now for each vector $w = (w_l) \in \mathbb{C}^{2n}$ consider the following functions (see Figure \ref{function_beta}):
$$ \beta^{(w)} (x) = \left\{   \begin{array}{ll}
v_{2\alpha -1} + s(w_{2\alpha -1} - v_{2\alpha -1})\,, & x = x^{(\alpha)}_{0} + s h_\alpha\,,  \\
w_{2\alpha -1}(1-s)\,, & x = x^{(\alpha)}_{1} + s h_\alpha\,, \\
s \,w_{2\alpha} \,, & x = x^{(\alpha)}_{r_\alpha -1} + s h_\alpha\,,  \\
w_{2\alpha} + s(v_{2\alpha} - w_{2\alpha}) \,, & x = x^{(\alpha)}_{r_\alpha} + s h_\alpha, \,,\\
0\,, & x_2^{(\alpha)} \leq x \leq x_{r_\alpha -1}^{(\alpha)} \,,
\end{array}\quad \begin{array}{c} 0 \leq s \leq 1\,,\\1 \leq \alpha \leq n\;.   \end{array}
\right. $$
Each function of the previous family is determined (apart from the vector $w$) by the vector $v = (v_l) \in \mathbb{C}^{2n}$ that collects the values of $\beta^{(w)}$ at the endpoints of the subintervals.    If we denote by $w^{(k)}$ the vectors such that $w^{(k)}_l = \delta_{lk}$, $k = 1, \ldots, 2n$, the $2n$ vectors $w^{(k)}$ are just the standard basis for $\mathbb{C}^{2n}$.   The corresponding functions $\beta^{(w)}$ will now be denoted simply by $\beta^{(k)}$. Notice that each boundary function $\beta^{(k)}$ is completely characterized by the unique nonextremal node where it does not vanish\footnote{If $k = 2\alpha-1$, it is the node $x^{(\alpha)}_1$, and if $k = 2\alpha$, it corresponds to the node $x^{(\alpha)}_{r_\alpha}$.}  and the values at the endpoints.  We denote by $v_l^{(k)}$, $l = 1, \ldots, 2n$, the boundary values of the functions $\beta^{(k)}$ above.

\subsubsection{The boundary matrix}
The $2n$ extremal values $v^{(k)}_{l}$ of the  boundary functions $\beta^{(k)}$ are undefined, but we are going to show that the $2n$ conditions \eqref{asorey} imposed on the boundary functions constitute a determinate system of linear equations for them. 
Because the boundary functions are constructed to be piecewise linear, the normal derivatives of these functions at the boundary can be obtained easily.   For the left boundaries of the intervals, i.e., at the points $a_\alpha$, we have

\begin{equation} \label{normal_left}
\left. \frac{d\beta^{(k)}}{d\nu}\right|_{x = a_{\alpha}} = - \frac{1}{h_\alpha}(w^{(k)}_{2\alpha-1} - v^{(k)}_{2\alpha -1}),
\end{equation}
and respectively, for the right boundaries,

\begin{equation}\label{normal_right}
\left. \frac{d\beta^{(k)}}{d\nu}\right|_{x=b_{\alpha}} = \frac{1}{h_\alpha} (v^{(k)}_{2\alpha} - w^{(k)}_{2\alpha}) = -\frac{1}{h_\alpha} (w^{(k)}_{2\alpha} - v^{(k)}_{2\alpha}) .
\end{equation}
Thus the vector containing the normal derivatives of the function $\beta^{(k)}$, consistently denoted by $\dot{\beta}^{(k)}$, is given by
\begin{equation}\label{derivatives}
\dot{\beta}^{(k)}_l = - \frac{1}{h_l}(w^{(k)}_l - v^{(k)}_l) =  - \frac{1}{h_l}(\delta_{lk} - v^{(k)}_l) ,
\end{equation}
where we use again the consistent notation $h_l = h_\alpha$, if $l = 2\alpha -1$, or $l = 2\alpha$.
For each boundary function $\beta^{(k)}$, the boundary conditions \eqref{asorey} read simply as the system of $2n$ equations on the components of the vector $v^{(k)}$,
$$ v^{(k)} - i \dot{\beta}^{(k)} = U(v^{(k)}  + i \dot{\beta}^{(k)} )  ,$$
or, componentwise,
$$ v^{(k)}_l \left(1-\frac{i}{h_l} \right) + \frac{i}{h_l} w^{(k)}_l = \sum_{j = 1}^{2n}U_{l j} \left[v^{(k)}_j \left(1 + \frac{i} {h_j}\right) - \frac{i}{h_j}w^{(k)}_j \right], \quad l = 1, \ldots, 2n .$$
Collecting coefficients and substituting the expressions $w^{(k)}_l=\delta_{lk}$ we get
\begin{equation}\label{boundaryvalues}
\sum_{j=1}^{2n}\left[\left(1-\frac{i}{h_j}\right)\delta_{lj}-U_{lj}\left(1+\frac{i}{h_j}\right) \right] v^{(k)}_{j} = \left[ -\frac{i}{h_k}\delta_{lk}-U_{lk}\frac{i}{h_k}\right].
\end{equation}
This last equation can be written as the matrix linear system:
\begin{equation}\label{boundary_equation}
 F V = C
 \end{equation}
with $V$ a $2n \times 2n$ matrix whose entries are given by $V_{jk} = v^{(k)}_j$, $j,k = 1, \ldots, 2n$. The $k$th column of $V$ contains the boundary values of the boundary function $\beta^{(k)}$.  The $2n\times 2n$ matrix $F$ with entries 
$$F_{lj} = \left(1-\frac{i}{h_j}\right)\delta_{lj}-U_{lj}\left(1+\frac{i}{h_j}\right) ,$$
will be called the boundary matrix  of the subdivision of the domain $M$ determined by the integer $N$,
and 
$$ C_{lk} = -\frac{i}{h_k} (\delta_{lk}  +  U_{lk}) $$
defines the inhomogeneous term of the linear system \eqref{boundary_equation}.
Using a compact notation we get
$$ F = \mathrm{diag} (\mathbf{1} - i/\mathbf{h}) - U \mathrm{diag} (\mathbf{1} + i/\mathbf{h}) , \quad \quad C = -i\,  (I  + U) \mathrm{diag}(1/\mathbf{h})  $$
where $1/\mathbf{h}$ denotes the vector whose components are $1/h_l$.
Notice that $F$ depends just on $U$ and the integer $N$ defining the size of the discretization leading to the approximate weak spectral problem defined below.
 
\subsection{Conditioning of the boundary matrix}\label{conditioningoftheboundarymatrix}
Before addressing the construction of the approximate spectral problem we will study the behavior of system \eqref{boundary_equation} under perturbations; in other words, we will compute the condition number of the boundary matrix $F$ and show that it is small enough to ensure the accuracy of the numerical determination of our family of boundary functions $\beta^{(i)}$.   The relative condition number we want to compute is
 $$ \kappa (F) = \norm{F}\norm{F^{-1}}.$$ 
In our case, the boundary matrix $F$ can be expressed as 
$$F = \bar{D} - U D = (I - U D \bar{D}^{-1}) \bar{D}$$ 
with $D_{jk} (\mathbf{h}) = D_{jk} = (1+\frac{i}{h_j})\delta_{jk}$.    Notice that the product $UD\bar{D}^{-1}$ is a unitary matrix which we will denote as $U_0(\mathbf{h})$ or simply $U_0$ if we do not want to emphasize the $\mathbf{h}$ dependence of $U_0$.  Thus, $F = (I - U_0)\bar{D}$ and
\begin{equation*}
\norm{F} = \norm{(I-U_0)\bar{D}} \leq \norm{I-U_0}  \norm{\bar{D}} \leq 2\norm{D}.
\end{equation*}
On the other hand, 
\begin{equation*}
\norm{F^{-1}} = \norm{\bar{D}^{-1}(I-U_0)^{-1}} \leq \norm{\bar{D}^{-1}}  \norm{(I-U_0)^{-1}} = \frac{\norm{D^{-1}}} {\min_{\lambda \in \operatorname{spec}(U_0)} \{ |1-\lambda| \} }
\end{equation*}
and thus we obtain
\begin{equation}\label{KF1}
 \kappa (F) \leq \kappa (D) \frac{2}{\min_{\lambda \in \operatorname{spec}(U_0)} \{ |1-\lambda| \} } .
 \end{equation}
As $D$ is a diagonal matrix its condition number is given by
$$\kappa (D) = \frac{\sqrt{\frac{1}{h^2_{\mathrm{min}}}+1}}{\sqrt{\frac{1}{h^2_{\mathrm{max}}}+1}} \leq \frac{h_{\max}}{h_{\min}}$$ 
with $h_{\max}$ ($h_{\min}$) the biggest (smallest) step of the discretization determined by $N$.  We get finally,
\begin{equation}\label{KF}
\kappa (F) \leq \frac{h_{\max}}{h_{\min}}\frac{2}{|1-\lambda|}
\end{equation}
with $\lambda$ the closest element of the spectrum of $U_0$ to $1$.   Of course, because $U_0$ is unitary, it may happen that $1$ is in its spectrum, so that the condition number is not bounded.  
Because the matrix $U_0$ depends on $\mathbf{h}$, its eigenvalues will depend on $\mathbf{h}$ too.   
We want to study the dependence of the closest eigenvalue to 1, or 1 for that matter, with respect to perturbations of the vector $\mathbf{h}$. \\

\begin{lemma}\label{perturbation_1}
Suppose that $X_0$ is an eigenvector with eigenvalue $1$ of $U_0$ and that the perturbed matrix $\hat{U} = U_0+\delta U$, for $\norm{\delta U}$ small enough, is such that $1 \in \sigma (\hat{U})$. Then $\bar{X}_0^T\delta U X_0 = 0$ to first order in $\delta U$.
\end{lemma}\\
\begin{proof} Clearly, if $1 \in \sigma (\hat{U})$ and $\delta U$ is small enough, there exist a vector $\hat{X} = X_0 + \delta X$,
with $\norm{\delta X} \leq C \norm{\delta U}$, such that $\hat{U} \hat{X} = 1 \hat{X}$.  Then we have
\begin{equation}
U_0\delta X +\delta U X_0+\delta U \delta X = \delta X\notag .
\end{equation}
Because $U_0 X_0 = X_0$ and $U_0$ is unitary, $\bar{X}_0^T U_0 = \bar{X}_0^T$, and then
multiplying on the left by $\bar{X}_0^T$ and keeping only first order terms, we get the desired condition:
$\bar{X}_0^T\delta U X_0 = 0.$
\end{proof}

\medskip

Because of the previous lemma, if $\hat{U} = U_0 + \delta U$ is a unitary perturbation of $U_0$ such that $\bar{X}_0^T\delta U X_0 \neq 0$, 
for any eigenvector $X_0$ with eigenvalue 1, then $1 \notin \sigma(\hat{U})$.   Now if we consider a unitary perturbation $\hat{U}$ of $U_0$ such that $1 \notin \sigma (\hat{U})$ we want to estimate how far 1 is from the spectrum of $\hat{U}$.
Consider the eigenvalue equation for the perturbed matrix.  The perturbed eigenvalue $\hat{\lambda} = 1 + \delta \lambda$ will satisfy
\begin{equation}\label{eigenvalueunitary}
(U_0+\delta U)(X_0+\delta X)=(1+\delta \lambda)(X_0+\delta X)\;.
\end{equation}
Multiplying on the left by $\bar{X}_0^T$ and solving for $|\delta\lambda|$ it follows that 
\begin{equation}\label{bounddelta}
|\delta \lambda|\geq \frac{|X^H_0\delta U X_0 + X^H_0\delta U \delta X|}{|1+X^H_0\delta X|}\geq  \frac{|X^H_0\delta U X_0|-|X^H_0\delta U \delta X|}{1+\norm{\delta X}} 
\end{equation}
for $\norm{\delta U}$ small enough.
Taking into account the particular form of the matrix $U_0 = U D \bar{D}^{-1}$ we have that $\delta U = U\delta(D\bar{D}^{-1})$ and therefore $\norm{\delta U}=\norm{\delta (D\bar{D}^{-1})}$.   Moreover, because $\delta(D\bar{D}^{-1})$ is a diagonal matrix,
$\delta(D\bar{D}^{-1})_{kk} = \frac{-2i}{ (h_k - i)^2}\delta h_k$ and its singular values are the modulus of its diagonal entries.  Hence we have the following proposition.\\

\begin{proposition}\label{lower_bound} Given the matrix $U_0 = UD\bar{D}^{-1}$ with eigenvalue 1, then 1 is not an eigenvalue of any unitarily perturbed matrix $U + \delta U = U_0 (D \bar{D}^{-1} +  \delta(D \bar{D}^{-1}))$ with $\norm{\delta(D\bar{D}^{-1})}$ small enough and $\delta h = \min\{ |\delta h_k| \} > 0$.  Moreover there exists a constant $C>0$ such that the perturbation $\delta \lambda$ of such eigenvalue satisfies the lower bound,
$$|\delta \lambda|\geq \frac{\sigma_{\min}(\delta (D\bar{D}^{-1}))-C\sigma^2_{\max}(\delta (D\bar{D}^{-1}))}{1+C\sigma_{\max}(\delta (D\bar{D}^{-1}))} > 0.$$
\end{proposition}\\

\begin{proof}
Because of Lemma \ref{perturbation_1}  it is sufficient to show that $\bar{X}_0^T \delta(D\bar{D}^{-1})) X_0 \neq 0$.
But this is an easy consequence of the fact that $|\sum_{i= 1}^{2n} (|X_{0,i}|^2 \delta(D\bar{D}^{-1}))_{ii}| \geq 2  \delta h$.   Furthermore there exists a constant $C>0$ such that $\norm{\delta X}\leq C\norm{\delta U}$; hence taking $\norm{\delta(D\bar{D}^{-1})}$ small enough we get $\sigma_{\min}(\delta (D\bar{D}^{-1}))-C\sigma^2_{\max}(\delta (D\bar{D}^{-1})) > 0$ and the bound follows from \eqref{bounddelta}.
\end{proof}

Now we can apply Proposition \ref{lower_bound} to (\ref{KF}) and if we neglect terms $|h_i^2|\ll 1$ and $|\delta h_i|\ll |h_i|$ we finally get the desired bound for the condition number
\begin{equation}\label{kappa_h}
\kappa (F)  \leq \frac{h_{\max}}{h_{\min}}\frac{1}{\delta h} .
\end{equation}
 Then, if for a given $N$ we obtain a boundary matrix $F$ which is bad conditioned, it suffices to change the size of the discretization, i.e., to increase $N$,  to improve the condition number.  Of course, if $N$ is already quite large, then the bound \eqref{kappa_h} could be useless.   For typical values $h\approx 10^{-2}\sim 10^{-3}$, it can be taken as $\delta h \approx 10^{-4} \sim 10^{-5}$ to provide condition numbers $\kappa(F) \approx 10^4\sim 10^5$.
 
\subsection{The spectral pencil}  For any $N > 2n$ we define the finite-dimensional approximation space $S^N$ as the linear span of the bulk and boundary functions, i.e., $S^N = \mathrm{span} \{ f_k^{(\alpha)}, \beta^{(l)} \mid \alpha = 1, \ldots, n, k = 2, \ldots, r_\alpha-1, l = 1, \ldots 2n  \}$.   All functions $f^{(\alpha)}_k$ and $\beta^{(l)}$ are linearly independent; thus the dimension of $S^N$ will be $|r| = r_1 + \ldots + r_n$.  It is convenient to rearrange the elements of the basis above as follows:
$$ \beta^{(1)}, f^{(1)}_1, \ldots, f^{(1)}_{r_1}, \beta^{(2)}, \beta^{(3)}, f^{(2)}_1, \ldots, f^{(2)}_{r_2}, \beta^{(4)}, \ldots, \beta^{(2n-1)}, f^{(n)}_1, \ldots, f^{(n)}_{r_n}, \beta^{(2n)} .$$
Using this ordering, we will rename the elements of this basis as $f_a$, with $a = 1, \ldots, |r|$, and an arbitrary element $\Phi_N \in S^N$ will be writen as $ \Phi_N (x) = \sum_{a=1}^{|r|}\Phi _a f_a(x) .$
We consider now the approximate eigenvalue problem (see \eqref{weak}):
\begin{equation}\label{approximate}
	\langle \Psi '_N,\Phi '_N\rangle - [\overline{\Psi }_N \Phi '_N]_{\partial M} + \langle \Psi _N, \mathcal{V}\Phi _N\rangle - \lambda_N\langle \Psi _N,\Phi _N \rangle = 0\quad \forall \Psi _N \in S^N.
\end{equation}
Introducing the expansion above in \eqref{approximate} we get
\begin{equation}\label{spectral_bilinear}
\sum_{a,b}^{|r|}\Psi _a \Bigl[\langle f'_a(x), f'_b(x)\rangle - [ \bar{f}_a(x) f'_b(x)]_{\partial M} +
\langle f_a(x), \mathcal{V}f_b(x)\rangle - \lambda \langle f_a(x), f_b(x) \rangle \Bigr] \Phi _b = 0 . 
\end{equation}
As \eqref{spectral_bilinear} holds for every $\Psi _N\in S^N$, this equation is equivalent to the eigenvalue equation of the matrix pencil $A-\lambda B$
\begin{equation}\label{pencil}
A=\lambda B,
\end{equation}
where $A_{ab} = \langle f'_a(x) , f'_b(x) \rangle - [\bar{f}_a(x) f'_b(x)]_{\partial M} + \langle f_a(x), \mathcal{V}f_b(x) \rangle$ and $B_{ab} = \langle f_a(x), f_b(x) \rangle$.   Notice that $A$ and $B$ are both Hermitian matrices, which improves the numerical algorithms used to compute the eigenvalues of the pencil.
In fact, when solving numerically \eqref{pencil}, it is relevant to preserve its Hermitian character. Notice that
the boundary functions $\beta^{(l)}(x)$ satisfy 
\begin{equation}\label{constraint}
[\bar{\beta}^{(l)} \beta^{(m)}{'} ]_{\partial M} - [\beta^{(m)} \bar{\beta}^{(l)}{'}]_{\partial M} = 0 ,
\end{equation}
because their boundary values are elements of a maximally isotropic subspace of the Lagrange boundary form.  Using \eqref{derivatives} and the definition of the boundary values of the boundary functions codified in the matrix $V$ we have
$$[\bar{\beta}^{(l)} \beta^{(m)}{'} ]_{\partial M} = \sum_{k=1}^{2n}\frac{1}{h_{k}}\bar{V}_{kl}(V_{km}-\delta_{km})
= \sum_{k=1}^{2n}\frac{1}{h_k}\bar{V}_{kl}V_{km}-\frac{1}{h_m}\bar{V}_{ml}.$$ 
This identity together with \eqref{constraint} leads to 
\begin{equation}\label{constraint2}
 \frac{1}{h_j}\bar{V}_{jk}=\frac{1}{h_k}V_{kj}.
\end{equation}
The Hermitian relation \eqref{constraint2} is satisfied by the numerical solutions of \eqref{boundary_equation} up to roundoff errors and consequently the pencil \eqref{pencil} is Hermitian only up to numerical roundoff errors.   We will force the numerical solution of matrix $V$ to satisfy \eqref{constraint2} so that the Hermiticity of the pencil is preserved exactly.  This is convenient  because the algorithms for solving the general eigenvalue problem are much better behaved in the Hermitian case \cite{demmel}. 

To end this discussion we must realize that, with the basis $f_a$ for $S^N$ we have just constructed, the matrices $A$ and $B$ are almost tridiagonal and the unique elements different from zero, besides the tridiagonal ones, are those related to the matrix elements of the boundary functions.  In fact, we can consider a number of cases.  If the function $f_a$ is an interior bulk function, i.e., not corresponding to the node $x^{\alpha}_2$ or $x^{(\alpha)}_{r_\alpha -1}$, it is obvious that the only nontrivial inner products $\langle f_a, f_b \rangle$ and $\langle f'_a , f'_b \rangle$ will correspond to $b = a-1,a, a+1$.   If the function $f_a$ is an extreme bulk function, for instance, $f_2^{(\alpha)}$, then it has nontrivial inner products only with $\beta^{(2\alpha-1)}$ and $f_3^{(\alpha)}$, and if $f_a$ is now a boundary function $\beta^{(l)}$, then the only nonvanishing inner products will be with the other boundary functions and an extreme bulk function, namely, $f_2^{(\alpha)}$ if $l=2\alpha-1$ or $f_{r_\alpha-1}^{(\alpha)}$ if $l=2\alpha$.  All these considerations would help to improve the accuracy and performance of the numerical algorithms we use to compute the eigenvalues of the spectral pencil.
  
\subsection{Convergence of the numerical scheme}\label{convergenceofthenumericalscheme}

Now we discuss the convergence of the proposed numerical scheme. Consider first the following space of functions
	$$\D_U = \Bigl\{\Phi \in\mathcal{C}^\infty(M) \mid \varphi - i \dot{\varphi} = U(\varphi + i \dot{\varphi}),\, \varphi=\Phi \bigr|_{\partial M},\,\dot{\varphi}=\frac{d\Phi }{d\nu}\bigr|_{\pM} \Bigr\}.$$
Then we have the following completeness result.\\

\begin{theorem}\label{RG}
The union of all the finite-dimensional spaces $S^N$ is the closure of  $\D_U$ in the Sobolev norm of order 1, 
$$\overline{ \displaystyle{\cup_{N>1} S^N} }^{\norm{\cdot}_{\H^1}}=\overline{\D_U}^{\norm{\cdot}_{\H^1}}.$$
\end{theorem}

\begin{proof} Let $\{M^N\}$, $0\leq\frac{1}{N}\leq 1$, be the nondegenerate family of subdivisions of the manifold $M\subset\mathbb{R}$ defined in section \ref{finiteelementsforgeneralsabc}. It consists of a collection of $|r|+n$ closed subintervals $I_{\alpha,a}$, $a=1,	\dots,r_\alpha+1$, of the real line. Each subinterval is a representation of the reference element $K=[0,1]$ with nodal set $\mathcal{N}=\{0,1\}$. According to subsection \ref{finiteelementsforgeneralsabc}, the family of functions $\mathcal{P}$ consists of the space of linear polynomials in $K$. Let $\mathcal{P}_m$ be the space of polynomials of degree $m$. Then the reference element $(K,\mathcal{P},\mathcal{N})$ satisfies the following for $m=1,2$\, and $l=0$:
\begin{itemize}
	\item $K$ is star-shaped with respect to some ball.
	\item $\mathcal{P}_{m-1}\subseteq \mathcal{P}\subseteq W^{m,\infty}(K)$.
	\item The nodal variables $\mathcal{N}$ involve derivatives up to order $l$.
\end{itemize}
For all $I_{\alpha,a}\in M^N$ let $(K,\mathcal{P}_{\alpha,a},\mathcal{N}_{\alpha,a})$ be the affine-equivalent element.   Let $P_N$ be the orthogonal projection operator onto $S^N$ and suppose that $1<p<\infty$ and $m-l-1/p>0$. Then, according to  \cite[Theorem 4.4.20]{Bre08}, there exists a constant $C$, depending on the reference element, $m$ and $p$ such that for $0\leq s\leq m$,
\begin{equation}\label{inequality_gen}
	\left(\sum_{I_{\alpha,a}\in M^N}\norm{\Psi-P_N\Psi}^p_{\H^{s}(I_{\alpha,a})}\right)^{1/p}\leq C N^{s-m}\norm{\Psi}_{\H^{m}(M)}.
\end{equation}
Particularizing for the case $s=1$, $m=2$, $p=2$, inequality (\ref{inequality_gen}) becomes
\begin{equation}\label{bound}
	\norm{\Psi -P_N\Psi}_{\mathcal{H}^1(M)}\leq \frac{C}{N}\norm{\Psi}_{\mathcal{H}^2(M)} .
\end{equation}
The inequality above and the boundary conditions in $\D_U$ and $S^N$ ensure that for every $\Psi\in \D_U$ there exists a sequence $\{\Psi_N\}_{N=0}^\infty$ with $\Psi_N\in S_N$ such that $\lim_{N\to\infty}\norm{\Psi-\Psi_N}_{\H^1(M)}=0$ and therefore $$\overline{\D_U}^{\norm{\cdot}_{\H^1}}\subset \overline{\cup_{N>1} S^N}^{\norm{\cdot}_{\H^1}}.$$

To prove the converse inclusion it is enough to realize that continuous functions are always contained in $\H^1(M)$. In fact, Sobolev's embedding theorem (cf. \cite{Ad75}) establishes, for $M$ of dimension 1, that one can always select a continuous representative for each equivalence class in $\H^1(M)$. Continuous functions verifying the boundary condition in \eqref{asorey} are therefore in $\smash{\overline{\D_U}}^{\norm{\cdot}_{\H^1}}$, and so are the finite-dimensional approximation spaces $S^N$ for any $N$.
\end{proof}\\

Let us now introduce the following notation for the bilinear form on the left-hand side of the first equation \eqref{weak}:
\begin{equation}\label{quadraticform}
	Q(\Psi,\Phi):=\langle \Psi' ,\Phi'\rangle-[\overline{\Psi}\Phi']\bigr|_{\partial M}+\langle \Psi ,\mathcal{V}\Phi\rangle.
\end{equation}
Notice that the bilinear form $Q$ is Hermitian, $Q(\Psi,\Phi)=\overline{Q(\Phi,\Psi)}$, iff the boundary conditions \eqref{asorey} are satisfied.\\

\begin{lemma}\label{Qsemibounded}  The bilinear form $Q$ \eqref{quadraticform} with domain $\smash{\overline{\D_U}}^{\norm{\cdot}_{\H^1}}$ is bounded below with respect to the Sobolev norm of order 1, i.e., it exists $C\geq0$ such that $$Q(\Phi,\Phi)\geq -C \norm{\Phi}^2_{\H^1}. $$
\end{lemma}

\begin{proof}
The potential function $\mathcal{V}$ is bounded below by hypothesis and obviously $\scalar{\Phi'}{\Phi'}\geq0$, so it suffices to prove that  $|[\overline{\Psi}\Phi']\bigr|_{\pM}|\leq C \norm{\Phi}^2_{\H^1}$.  Notice that we can diagonalize the unitary matrix $U$ as
$$ U = S^\dagger D S ,$$
where $D = \mathrm{diag}  (-1,\ldots,-1,z_1,\ldots, z_r)$ with $z_i$ complex numbers of modulus 1 different from $-1$.
Now notice that the boundary condition \eqref{asorey} can be rewritten taking into account the diagonalization of $U$ as $(I-D)S\varphi=i(I+D)S\dot{\varphi}$ or equivalently 
$$	\begin{bmatrix}
		2 & \\ & Z_-
	\end{bmatrix}
	\begin{bmatrix}
		\varphi_0 \\ \varphi_1 
	\end{bmatrix}=
	i\begin{bmatrix}
		0 & \\ & Z_+
	\end{bmatrix}
	\begin{bmatrix}
		\dot{\varphi}_0 \\ \dot{\varphi}_1 
	\end{bmatrix},$$ 
 where we have defined $Z_\pm:=\mathrm{diag}(1\pm z_1,\dots,1\pm z_r)$ and $$\begin{bmatrix}\varphi_0 \\ \varphi_1\end{bmatrix}:=S\varphi,\quad \begin{bmatrix}\dot{\varphi}_0 \\ \dot{\varphi}_1\end{bmatrix}:=S\dot{\varphi}$$ are just the boundary data adapted to the new basis of $\mathbb{C}^{2n}$. Note that this system of equations now implies that $\mathbb{C}^{2n-r}\ni \varphi_0=0$ and $\mathbb{C}^{r}\ni \dot{\varphi}_1=-iZ_+^{-1}Z_-\varphi_1$. Moreover, as $S$ is a unitary operator in the Hilbert space of the boundary, the inner product $\langle \varphi, \dot{\varphi} \rangle$ verifies 
$$  \langle \varphi, \dot{\varphi} \rangle = -i \langle \varphi_1 , Z_+^{-1}Z_- \varphi_1 \rangle  ,$$
and then
\begin{equation}\label{acotacioncuadratica}
|  \langle \varphi, \dot{\varphi} \rangle | \leq \norm{Z_- Z_+^{-1}} \norm{ \varphi}^2 \leq c \norm{\Phi}_{\H^1}^2,
\end{equation}
where the last inequality is a particular case of the boundary trace inequalities \cite{Ad75} taking $\varphi\in L^2(\pM)$. 
It can happen that $\{-1\}\not\in \sigma(U)$; then $(I+U)$ is invertible and \eqref{acotacioncuadratica} still holds with $Z_{\pm}=I\pm U$.
Taking into account that the potential $\mathcal{V}$ is assumed to be bounded below and that $\norm{\cdot}\leq\norm{\cdot}_{\H^1}$, it follows that the quadratic form is bounded below with respect to the Sobolev norm of order 1.
\end{proof}\\ 

\begin{corollary}\label{positive definite}
Let $U\in U(2n)$ be a unitary matrix whose spectrum $\sigma(U)\subseteq \{-1,1\}$. Then, the quadratic form Q is bounded below.
\end{corollary}\\

\begin{proof} If the spectrum of the unitary $\sigma(U)\subseteq\{-1,1\}$, then the matrix $Z_-$ of the  theorem above is the zero matrix. Therefore \eqref{acotacioncuadratica} becomes $\scalar{\varphi}{\dot{\varphi}}=0$ and $$Q(\Phi,\Phi)\geq k\norm{\Phi}^2\;,$$ where $k$ is the semibound of the potential $\mathcal{V}$.
\end{proof}\\

We will end the proof of the convergence of the algorithm by using the variational characterization of eigenvalues.  However, in order to do that  the quadratic form $Q$ has to be bounded below not only in $\mathcal{H}^1(M)$ but in $L^2(M)$ for any unitary $U\in\mathcal{U}(2n)$.    We will discuss this in the next paragraphs.

Let us recall first that the class of Schr\"odinger operators considered here are self-adjoint extensions of closed symmetric operators with domain $\mathcal{H}^2_0(M)$ (cf. Section \ref{Selfadjoint_extensions_of_Schrodinger_operators}). Notice that even if the symmetric operators defined in $\mathcal{H}^2_0$ are bounded below by $-\norm{\mathcal{V}}_{\infty}$, the self-adjoint extensions defined in Eq.\ \eqref{schrodinger} need a more careful analysis. This is due to the boundary term $-[\bar{\Psi\Phi^\prime}]\bigr|_{\partial M}=\scalar{\psi}{\dot{\varphi}}_{\mathcal{L}^2(\partial M)}$, which is not positive defined, that appears in the weak form \eqref{weak} of the problem.

Because the manifold $M$ is compact and of dimension 1,  the deficiency indices of the symmetric operators are finite and equal to $2n$, the number of elements in the boundary of the manifold $\partial M$. Notice that the ordinary differential equation of order 2 defining the deficiency spaces (\ref{deficiency}) has two independent solutions on each subinterval.  Therefore, according to \cite[Theorem 8.18]{We80}, all self-adjoint extensions of this symmetric operators have the same essential spectrum and because the Dirichlet extension has a pure discrete spectrum, then all the self-adjoint extensions do so.

Accordingly, let us select a self-adjoint extension of the Schr\"odinger operator $H$ defined by the unitary matrix $U$.   We denote by $D_U$ its domain and by $H_U$ the self-adjoint operator corresponding to such unitary matrix.  Notice that because of Theorem \ref{self1} and (\ref{graph_uni}),  $\gamma(D_U) =  \mathrm{Gr} (A_U ) = C^{-1}(\mathrm{Gr}(U))$, where $\mathrm{Gr} (U)$ denotes the graph of the linear map $U$ and $A_U = i(U-I)/(U+I)$ denotes the Cayley transform of $U$ (we assume that $-1 \notin \sigma (U)$).  Then $Q(\Phi, \Psi) = \langle \Phi , H_U \Psi \rangle$ for any $\Phi,\Psi \in D_U$  (notice that $D_U \subset \smash{\overline{\D_U}}^{\norm{\cdot}_{\H^1}}$).     Finally, the trace map $\gamma$ defines a continuous isomorphism between the subspace $(I+K)\mathcal{N}_+ = \mathrm{Gr} (K)$ appearing in the definition of the domain of the self-adjoint extension $H_U$, Eq.\ (\ref{Neumanndomain}), and the subspace $\mathrm{Gr}(A_U)$.  Thus the unitary matrix $K$ is implicitly defined as a function of $U$ alone.   

This correspondence can be made more explicit as follows.   Let $U$ be a unitary matrix defining a self-adjoint extension of $H$ such that $-1\notin \sigma(U)$.  Now consider a boundary value function $\varphi$ (that in our case is just a vector in $\mathbb{C}^{2n}$),  then there is exactly one solution to each of the boundary data problems $H_{\alpha} \Phi_\alpha = \pm i \Phi_\alpha$, $\Phi_\alpha\mid_{I_\alpha} = \varphi_\alpha$, where $H_\alpha = H\mid_{I_\alpha}$ on each subinterval $I_\alpha$, Eq.\ (\ref{sturm}), and $\varphi_\alpha$ denotes the restriction of $\varphi$ to the interval $I_\alpha$.   Let us denote such solutions by $\xi^{(\alpha)}_\pm$ respectively and by $\xi_\pm$ the vector whose components are $\xi_\pm^{(\alpha)}$.  Then we choose a combination $a\circ \xi_+ b\circ \xi_-$, with $a,b$ $n$-component vectors and $\circ$ the pointwise Hadamard product, such that $\gamma (a \circ \xi_+ + b\circ \xi_-) = (\varphi, A_U(\varphi))$.  Notice that there is just one such combination.   Then we define $K(a\circ \xi_+) = b\circ \xi_-$.   Then, the matrix $K$ will depend on the Wronskian of the fundamental system of solutions provided by $\xi_\pm^{(\alpha)}$ and the inhomogeneous part of the previous equation, thus it is clear that the matrix $K$ depends continuously on $A_U$.   

Let us consider now a continuous map $\gamma \colon [0,1] \to U(2n)$ into the space of unitary maps such that $-1 \notin \sigma(U(t))$ for all $t$ and we denote $U(t) = \gamma (t)$.  Then, it is clear that,  if we denote by $K(t)$ the unitary map corresponding to $A_{U(t)}$, it depends continuously on $t$ (again, the graph of $K(t)$ is composed of solutions of the deficiency equations $H_{U(t)} \xi_{\pm} = \pm i \xi_{\pm}$ satisfying the normal boundary conditions $\dot{\varphi} = A_{U(t)} \varphi$, and they depend continuously on $t$).    The eigenvalues and eigenfunctions of an elliptic differential operator depend continuously on the parameters of the problem, but because Eq.\ (\ref{neumann_op}) $H_{U(t)} = H_0 \oplus i(I-K(t))$, thus we may conclude that if $U(t)$ is a continuous path of unitary operators not containing $-1$ in the spectrum, the eigenvalues $\lambda (t)$ of the corresponding self-adjoint extensions will depend continuously on $t$ .   \\

\begin{theorem} \label{L2semibounded} Let $U\in U(2n)$ be a unitary matrix. Then, the quadratic form $Q$ defined in the domain $\smash{\overline{\D_U}}^{\norm{\cdot}_{\H^1}}$ corresponding to the self-adjoint extension $H_U$ of the symmetric Schr\"odinger operator $H$ is bounded below in $L^2(M)$.
\end{theorem}\\

\begin{proof}   Let $S$ be a unitary matrix diagonalizing $U$, i.e., $S^\dagger U S = \mathrm{diag} (e^{i\alpha_1}, \ldots, e^{i\alpha_{2n}})$.  Let the path $\gamma \colon [0,1] \to U(2n)$ be defined as $\gamma(t) = S D(t) S^\dagger$ with $D(t)$ the diagonal matrix whose entries are given by $D_{kk} (t) = e^{it\alpha_k}$ if $\alpha_k < \pi$, $D_{kk} (t) = e^{i(t\alpha_k +(1 - t)2\pi)}$ if $\alpha_k > \pi$ and $D_{kk} (t) = -1$ if $\alpha_k = \pi$ .  Then it is clear that $U(0)$ is diagonal with $\pm 1$ on it depending on the element.   On the other hand $U(1) = U$.    The eigenvalues of the matrix $U(t)$ will never cross $-1$ (if they are $-1$, they remain constant) and because of the continuous dependence of eigenvalues on $t\in [0,1]$, their variation is bounded.  
Finally notice that the form $Q$ is bounded below on the domain $\overline{\D_U(t)}^{\norm{\cdot}_{\H^1}}$ iff the corresponding Hamiltonian $H_{U(t)}$ is bounded below. Now because of corollary \ref{positive definite}, the quadratic form associated to $U(0)$ and therefore the Hamiltonian $H_{U(0)}$ is bounded below.   Notice that the ellipticity of the Schr\"odinger operator $H_U$ guarantees that only a finite number of eigenvalues can become negative and therefore any self-adjoint extension in the path $\gamma$ will be bounded below.
Notice that if the path $\gamma(t)$ is smooth, then a given eigenvalue $\lambda(t)$ satisfies the differential equation $\dot{\lambda}(t) = \langle \Phi(t) , \dot{K(t)}  \Phi(t) \rangle$, with $\Phi(t)$ the eigenvector of $\lambda(t)$, but because of the differentiability of $K(t)$ and the compactness of the interval $[0,1]$ the total variation of $\lambda (t)$ is bounded by $C || K(t) ||_\infty$, i.e., independently of $\lambda$.
\end{proof}\\

Now the convergence of the eigenvalues and eigenvectors follows by standard arguments.\\

\begin{theorem}\label{ConvergenciaSoluciones}
The solutions $(\Phi_N,\lambda_N)$ of the approximate eigenvalue problem \eqref{approximate} converge in the limit $N\to \infty$ to solutions $(\Phi,\lambda)$ of the weak problem \eqref{weak}.
\end{theorem}\\

\begin{proof}
Theorem \ref{L2semibounded} assures that the \textrm{min-max} principle \cite{Re75} applies for the quadratic form \eqref{quadraticform} and the convergence of the eigenvalues is proved as follows.

Applying the \textrm{min-max} principle to the quadratic forms \eqref{quadraticform} of the weak and the approximate problem and substracting them we get
\begin{equation}
\sup_{\varphi_1,...,\varphi_{n-1}}\biggl[\inf_{\stackrel{\Phi _N\in[\varphi_1,...,\varphi_{n-1}]^\bot}{\Phi _N\in S^N}} \frac{Q(\Phi _N,\Phi _N)}{\norm{\Phi _N}^2}\biggr]-\sup_{\varphi_1,...,\varphi_{n-1}}\biggl[\inf_{\stackrel{\Phi \in[\varphi_1,...,\varphi_{n-1}]^\bot}{\Phi \in \D_U}} \frac{Q(\Phi ,\Phi )}{\norm{\Phi}^2}\biggr]=\lambda_N-\lambda.
\end{equation}
The left-hand side tends to zero in the limit $N\to\infty$ by Theorem \ref{RG} and therefore
\begin{equation}\label{Convergenceeigenvalue}
	\lim_{N\to\infty}\lambda_N=\lambda.
\end{equation}
Now let $\Phi_N\in \H^1(M)\subset L^2(M)$ be a solution of \eqref{weak}, i.e. $$Q(\Psi_N,\Phi_N)-\lambda_N\scalar{\Psi_N}{\Phi_N}=0,\quad \forall \Psi_N\in S^N.$$  We can assume that $\norm{\Phi_N}_{\H^1}=1$ for every $N$ and therefore the Banach-Alaoglu theorem ensures that subsequences $\{\Phi_{N_j}\}$ and accumulation points $\Phi$ exist such that for every $\Psi\in\smash{\overline{\D_U}}^{\norm{\cdot}_{\H^1}}$ 
\begin{equation}\label{BanachAlaoglu}
\lim_{N_j\to\infty}|\scalar{\Psi}{\Phi-\Phi_{N_j}}_{\H^1}|=0.
\end{equation}
If we denote by $P_N\Psi$ the orthogonal projection of $\Psi$ into $S^N$, we have that for every $\Psi\in\smash{\overline{\D_U}}^{\norm{\cdot}_{\H^1}}$
\begin{align*}
	|Q(\Psi,\Phi)-\lambda\scalar{\Psi}{\Phi}|&\leq|Q(\Psi,\Phi-\Phi_{N_j})|+|\lambda||\scalar{\Psi}{\Phi-\Phi_{N_j}}|+|\lambda_{N_j}-\lambda||\scalar{\Psi}{\Phi_{N_j}}|+\cdots\\
	&\phantom{\leq}\cdots+ |Q(\Psi-P_N\Psi,\Phi_{N_j})|+|\lambda_{N_j}||\scalar{\Psi-P_N\Psi}{\Phi_{N_j}}|.
\end{align*}
Now \eqref{bound}, \eqref{Convergenceeigenvalue}, \eqref{BanachAlaoglu}, and Lemma \ref{Qsemibounded} ensure that the right-hand side vanishes in the limit $N_j\to \infty$ and we get that the accumulation points of the sequence $\{\Phi_N\}$ are eigenvectors of the weak eigenvalue problem.
\end{proof}

%%%%%%%%%%%%%%%%%%%%%%%%%%%%%%%%%%%%%%%%%%%%%%%%%%%%%%%%%%%%%%%%%%%%%%%%%%%%%%%%%%%%%%%%%%%%%%%%%%%%

\section{Numerical experiments}\label{Numerical_experiments_and_conclusions}

%%%%%%%%%%%%%%%%%%%%%%%%%%%%%%%%%%%%%%%%%%%%%%%%%%%%%%%%%%%%%%%%%%%%%%%%%%%%%%%%%%%%%%%%%%%%%%%%%%%%%%%%%%%%%%%%

The numerical scheme described in section \ref{FEM} results in the finite dimensional eigenvalue problem of \eqref{pencil}.  Furthermore, we have the $1/N$ bound \eqref{bound} on the discretization error.   Hence, if we were able to solve \eqref{pencil} for increasing grid size $N$, we would get better and better approximations to the eigenvalue problem. As remarked in the previous section, the pencil is almost tridiagonal and both matrices $A$ and $B$ are Hermitian.  Hence, the resulting problem is algebraically well behaved and it should lead to accurate results.  

We will now discuss some numerical experiments that illustrate the stability and the convergence of the algorithm.  We will also compare these results with those obtained by using two other algorithms (not based on the finite element method). In the latter cases we will use a particular choice of boundary conditions close to the singular case described in \cite{Be08}.  
We will consider the spinless free particle in $M=[0,2\pi]$ ($\mathcal{V}(x) = 0$) subject to different boundary conditions:
\begin{equation}\label{spinlessfree}
-\frac{d^2}{dx^2}\Psi=\lambda \Psi\;.
\end{equation}
Notice that in this simple case there is only one interval and $r=N+1$, so we could use $r$ or $N$ interchangeably.  The number $h=\frac{2\pi}{r+1}$ is going to be the length of each subinterval.
After some straightforward computations we find that the matrices $A$ and $B$ defining the spectral pencil $A -\lambda B$ associated to (\ref{spinlessfree}) are given by

%%%%%%%%%%%%%%%%%%%%%%%%%%%%%%%%%%%%%%%%%%%%%%%
\begin{equation}
		\begin{array}{ll}
		
			\begin{array}{l} 
				a_{kk}=2,\quad k=2,\dots, r-1\,,\\
				a_{k \, k+1}=a_{k+1\, k}=-1,\quad k=1,\dots,r-1\,,\\
				a_{11}=2-V_{11}\,,\\
				a_{rr}=2-V_{22}\,,\\
				a_{1r}=-V_{12}\,,\\
				a_{r1}=\overline{a_{1r}}=-V_{21}\,,\\
				a_{kj}=0,\quad j\geq k+2,\quad j\leq k-2\,,\\
				A_{kj}=\frac{1}{h}a_{kj}\,,
			\end{array}  &
%%%%%%%%%%%%%%%%%%%%%%%%%%%%%%%%%%%%%%%%%%%%%%%
			\begin{array}{l}
				b_{kk}=4,\quad k=2,\dots, r-1\,,\\
				b_{k\, k+1}=b_{k+1\, k}=1,\quad k=1,\dots,r-1\,,\\
				b_{11}=4+2[|V_{11}|^2+|V_{21}|^2]+2V_{11}\,,\\
				b_{rr}=4+2[|V_{22}|^2+|V_{12}|^2]+2V_{22}\,,\\
				b_{1r}=2[\bar{V}_{11}V_{12}+\bar{V}_{21}V_{22}]+2\,,\\
				b_{r1}=\overline{b_{1r}}\,,\\
				b_{kj}=0,\quad j\geq k+2,\quad j\leq k-2\,,\\
				B_{kj}=\frac{h}{6}b_{kj}\;.
			\end{array}
			
		\end{array}		
\end{equation}

In each case, in order to obtain the matrix $V$, one needs to solve previously the corresponding system of equations \eqref{boundary_equation} for the given self-adjoint extension. The solutions of these generalized eigenvalue problems have been obtained using the \textrm{Octave} built-in function $\mathbf{eig}$. The details of this routine can be found in \cite{An92}. This built-in function is a general-purpose diagonalization routine that does not exploit the particularly simple algebraic structure of this problem. One could adapt the diagonalization routine to the algebraic structure of the problem  at hand (there are only two elements outside the main diagonals) to improve the efficiency. Moreover, one could also use the $p$-version of the finite element method that, considering that the solutions are smooth, would be more adaptive.   However our main objective here is to show that the computation of general self-adjoint extensions by using nonlocalized finite elements at the boundary, as explained in subsection \ref{sectionboundaryfunctions}, is reliable and accurate.  We consider that the results explained below account for this, and we leave these improvements for future work.

First we will test the stability of the method against variations of the input parameters.  The parameters of this procedure are the matrix $U$ determining the self-adjoint extension whose eigenvalue problem we want to solve.  We will perturb an initial self-adjoint extension, described by a unitary matrix $U$, and we will observe the behavior of the eigenvalues. In other words, we are interested now in studying the relation 
				\begin{equation}
					|\Delta \lambda|=K(\varepsilon)\norm{\Delta U}\;,
				\end{equation}
where $\epsilon$ is the parameter measuring the size of the perturbation. If the algorithm were stable one would expect that the condition number $K(\varepsilon)$ would grow at most polynomially with the perturbation $\epsilon$.  However, we must be careful in doing so since the exact eigenvalue problem presents divergences under certain circumstances (explained below) which could lead to wrong conclusions.  In fact, as a consequence of Lemma \ref{Qsemibounded}, we see that when a self-adjoint extension is parameterized by a unitary matrix $U$ that has eigenvalues close but not equal to $-1$, it happens that some eigenvalues of the considered problem take very large negative values. However, matrices with $-1$ in the spectrum can lead to self-adjoint extensions that are positive definite, for example, Dirichlet or Neumann self-adjoint extensions.  Thus, following a path in the space of self-adjoint extensions, it could happen that a very small change in the arc parameter leads to an extremely large jump in the exact eigenvalues.  Such self-adjoint extensions are precisely the ones that lead in higher dimensions to the problem identified by M.~Berry as a Dirichlet singularity\footnote{Notice again that $U=-\mathrm{1}$ is the unitary matrix describing Dirichlet boundary conditions.} \cite{Be08},\cite{Be09}, \cite{Ma09} and they will be the target of our latter tests. For proving the stability it is therefore necessary to perturb the unitary matrix along a direction of its tangent space such that this situation is avoided.  A path in the space of self-adjoint extensions where these jumps do not occur is, for instance, the one described by the so-called quasi-periodic boundary conditions \cite{As83}.  In this case, the self-adjoint domain is described by functions that satisfy the boundary conditions $\Psi(0)=e^{i2\pi\epsilon}\Psi(2\pi)$ and $\Psi'(0)=e^{i2\pi\epsilon}\Psi'(2\pi)$, which correspond to the unitary matrix 
\begin{equation}\label{unitariaquasiperiodicas}
	U(\epsilon)=\begin{pmatrix} 0& e^{i2\pi\epsilon}\\e^{-i2\pi\epsilon}&0 \end{pmatrix}.
\end{equation}
Notice that this particular choice of boundary conditions, which are nonlocal in the sense that they mix the boundary data at both endpoints of the interval, can naturally be treated by the discretization procedure introduced in section \ref{FEM} and that goes beyond the ones usually addressed by most approximate methods \cite{Ac09}, \cite{Ch99}, \cite{Fa57}, \cite{Pr73}, described by local equations of the form $\alpha\cdot \Psi(a)+\beta\cdot \Psi'(a)=0$, $\gamma\cdot \Psi(b)+\delta\cdot \Psi'(b)=0$.

Let us then consider perturbations of the periodic case in the quasi-periodic direction, i.e., we consider 
$$U(\varepsilon)\simeq U+i\varepsilon A=\begin{pmatrix} 0 & 1\\1 & 0 \end{pmatrix}+i\varepsilon \begin{pmatrix} 0 & 1\\ -1 & 0\end{pmatrix}.$$  We have calculated the numerical solutions for values of $\varepsilon$ between $10^{-4}$ and $10^{-1}$ in steps of $10^{-4}$ and the discretization size used was $N=250$. In the latter case the perturbation is $\norm{\Delta U}=\norm{i\varepsilon A}=\varepsilon$, hence the absolute error ratio is $$K(\varepsilon)=\frac{|\Delta \lambda|}{\norm{\Delta U}}=\frac{1}{\varepsilon}|\Delta \lambda|.$$ The results are plotted in figure \ref{estabilidad}. 
\begin{figure}[h]\centering
\includegraphics[height=3in]{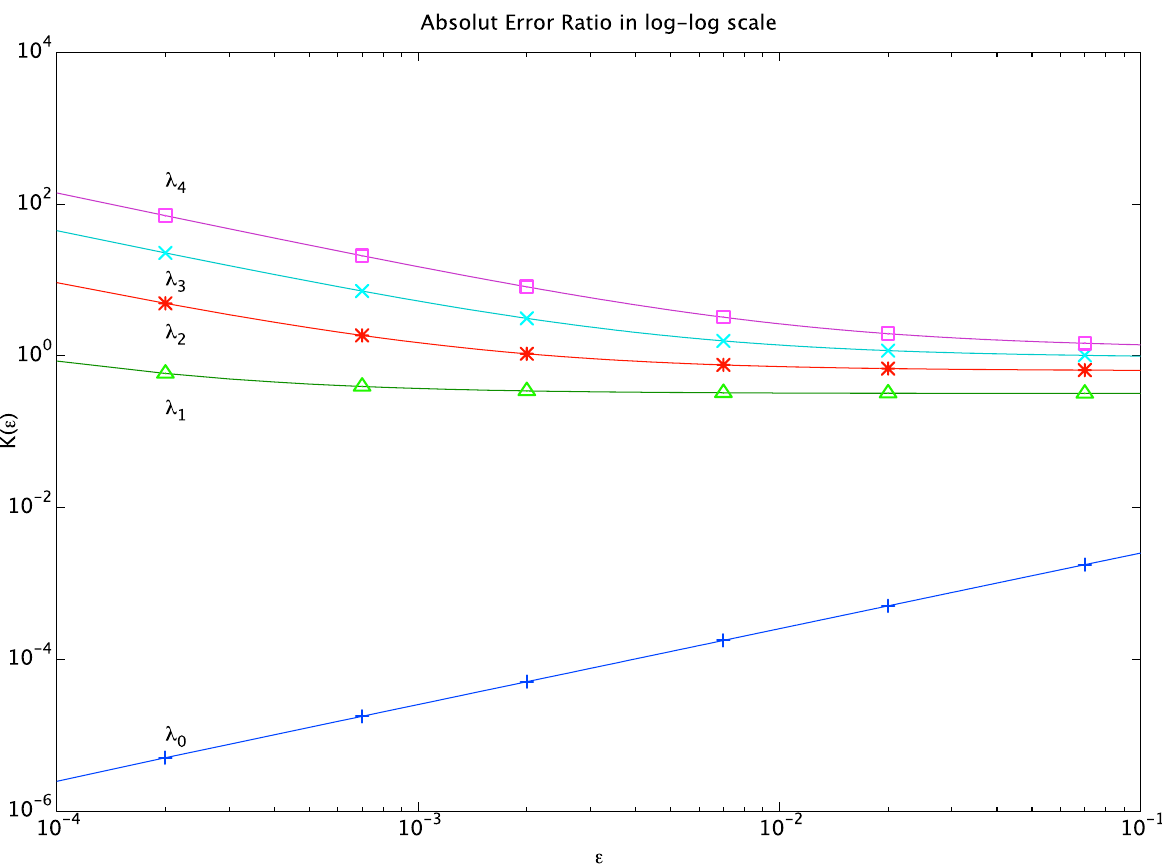}
\caption{Absolute error ratio $K(\varepsilon)$ of the five lowest levels for the periodic boundary problem plotted against $\varepsilon$ in log-log scale.}\label{estabilidad}
\end{figure}
As can be seen, $K(\varepsilon)$ is a decreasing function of $\varepsilon$ except for the fundamental level. In the latter case the absolute ratio is an increasing function of $\varepsilon$. However, in this case it can be seen clearly that the growth is linear. This result shows that the procedure is stable under perturbations of the input matrix $U$.\\

We will consider now the free particle \eqref{spinlessfree} subjected to quasi-periodic boundary conditions:
\begin{equation}
	\Psi(0)=e^{i2\pi\epsilon}\Psi(2\pi)\,, \quad
	\Psi'(0)=e^{i2\pi\epsilon}\Psi'(2\pi) \;.
\end{equation}
These are codified by the unitary matrix \eqref{unitariaquasiperiodicas}. This is a meaningful example since it demonstrates that this algorithm provides a new way to compute the Bloch decomposition of a periodic Schr\"odinger operator. This problem is addressed usually by considering the unitary equivalent problem $-\left( \frac{d}{d x}+i\epsilon \right)^2 \Psi = \lambda\Psi $ 
with periodic boundary conditions. However, our procedure is able to deal with it directly in terms of the original boundary condition.

The analytic solutions for this particular eigenvalue problem can be obtained explicitly \cite{As83}. They are $\lambda_n=(n+\epsilon)^2$, \smash{$\psi_n=\frac{1}{\sqrt{2\pi}}e^{-i(n+\epsilon)x}$}, $n=0,\,\pm1,\,\pm2,\dots$, and we can compare the approximate solution obtained by the procedure described in section \ref{FEM}. In particular we show that the bound \eqref{bound} is satisfied, i.e., that the error between the approximate solutions and the analytic ones measured in the $\mathcal{H}^1$-norm is of the form $1/N$. The results for the five lowest eigenvalues corresponding to $\epsilon=0.25$ are shown in figure \ref{convergencia}.\footnote{In order to avoid errors coming from numerical quadratures, the solutions of the integrals appearing in the $\mathcal{H}^1$-norms of each subinterval, comparing the analytic solutions and the linear approximations, have been computed explicitly.} As is known for the finite element approximations, the error grows with the order of the eigenvalue; however, it is clearly seen that for all the cases the decay law is of the form $1/N$, therefore satisfying the bound.\\

\begin{figure}[h]\centering
\includegraphics[height=4in]{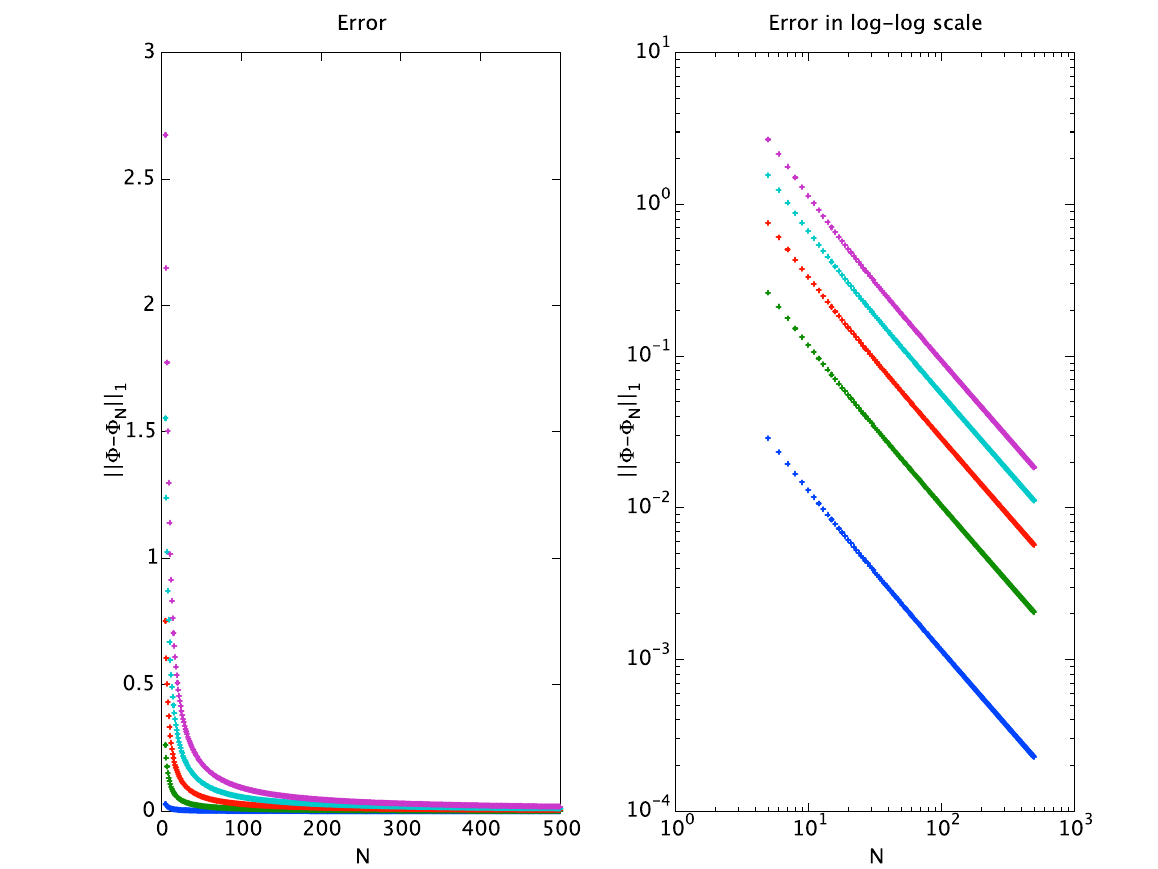}
\caption{Evolution of the error, measured in the $\mathcal{H}^1$-norm, for the five lowest levels of the quasi-periodic free particle problem ($\epsilon=0.25$) with increasing lattice size plotted in normal scale and log-log scale.}\label{convergencia}
\end{figure}

Finally we consider again the free particle, but in this case subjected to the local boundary conditions
\begin{equation}\label{condicionesDsingularity}
		\begin{array}{c}
			\Psi'(0)=0\,,\\
			\tan(-\theta/2)\cdot \Psi(2\pi)+\Psi'(2\pi)=0\;,
		\end{array}
\end{equation}
which are determined by the unitary matrix
\begin{equation}\label{unitariaDsingularity}
	U=\begin{pmatrix} 1 & 0 \\ 0 & e^{-i\theta} \end{pmatrix}\;.
\end{equation}
These boundary conditions, unlike the previous ones, can be handled by most of the software packages available for the integration of Sturm-Liouville problems.   When $\theta=\pi$ the free particle operator is positive, but for values of $\theta<\pi$ the fundamental level is negative with increasing absolute value as $\theta\to\pi$. It happens that for values of $\theta\simeq\pi$ the absolute value of this fundamental eigenvalue is several orders of magnitude bigger than the closest eigenvalue.

The solutions of the discrete problem, according to Theorem \ref{ConvergenciaSoluciones}, are guaranteed to converge to solutions of the exact problem for increasing lattice size. However, it is not necessary that the sequence of eigenvalues obtained in the approximate solution is in correspondence with the sequence of eigenvalues of the exact problem. It may happen that for some threshold $N$ some \emph{new} eigenvalues appear that were not detected for smaller $N$. The big gap between the two lowest levels in the self-adjoint extensions described above is a good example of this feature. We have computed the spectrum for a fixed value $\theta=0.997\pi$ for $N$ in increasing steps. For each value of $N$ we show the lowest five eigenvalues returned by the procedure. The results are plotted in Figure \ref{ConvergenciaAutovaores}. Notice that for $N\leq1300$ the negative fundamental level is not detected. However, for $N=1400$ an approximation of the negative eigenvalue is returned ($\sim-6000$). Now, in this situation and for $N>1400$, the second to fifth lowest eigenvalues coincide with the lowest four returned for the situations below the threshold. The convergence of the negative eigenvalue is also visible. The scale in the negative $y$-axis has been rescaled so that the performance could be better appreciated.
\begin{figure}[h]\centering
\includegraphics[height=3.75in]{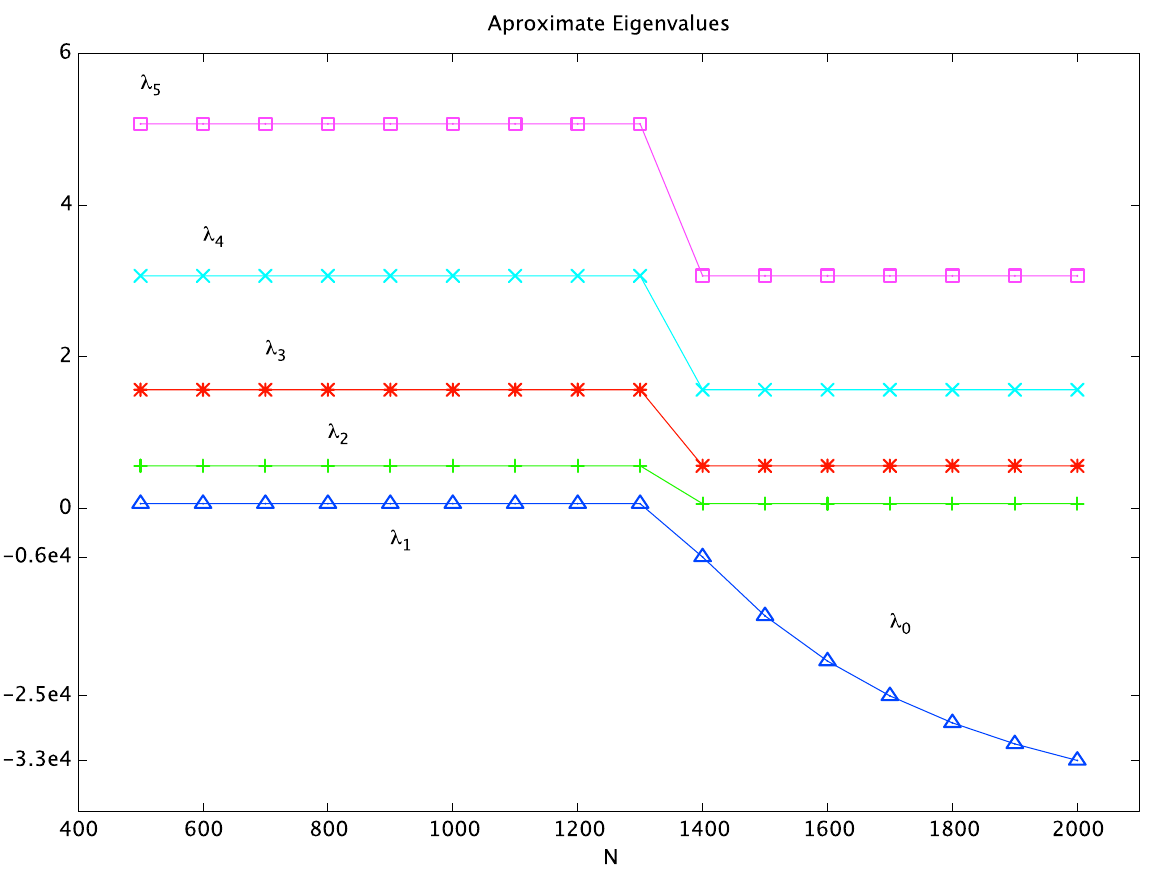}
\caption{Five lowest eigenvalues for $\theta=0.997\pi$ and increasing lattice size.}\label{ConvergenciaAutovaores}
\end{figure}
\\

As stated at the beginning of this section, we are now going to compare the results obtained with the numerical scheme described so far, from now on referred to as FEM, with two other algorithms that are not based on the finite element method, namely the line-based perturbation method (LPM) and the constant reference potential pertubation method (CPM) proposed in \cite{Le06} and \cite{Le05}, respectively. The software used for the calculations is, in the latter cases, a ready-to-run version provided by the authors.\footnote{http://www.twi.ugent.be.} The objective is now to find the solutions for the Schr\"odinger problem with the boundary conditions given by \eqref{condicionesDsingularity}.

The LPM routine reported errors for the lowest order eigenvalues for $\theta\geq 2.6$ and was not able to produce any output to compare. The results comparing the CPM and FEM routines for the free particle are shown in Figures \ref{laplaciano2} to \ref{laplaciano4}. The discretization size used for all the FEM calculations was $N=5000$. Both methods provide almost the same results for the first excited states (indices 1 to 5) of the free particle; however, this is not the case for the fundamental level (Figure \ref{laplaciano2}), where it is expected that the eigenvalues take increasingly large negative values for $\theta$ approaching $\pi$. Note that the lowest value achieved with FEM is $\lambda_0\simeq -7\cdot 10^4$. Clearly the CPM routine fails to reproduce the correct eigenvalues in this situation, although it provides good approximations for $\theta\leq 0.989\pi$. In all these cases the CPM routine issued warnings on the low reliability of the results and an initial estimate for the fundamental eigenvalue was necessary.  Although the numerical eigenvalues are in agreement with the expected ones, except for the fundamental level, the eigenfunctions plotted in figures \ref{laplaciano3} and \ref{laplaciano4} for the case $\theta=3.1$ show clearly that the solutions provided by FEM are more accurate. For instance, the solutions obtained with CPM present discontinuities near the middle of the interval, although they are expected to be smooth functions. The eigenfunction obtained for the fundamental level (Figure \ref{laplaciano4}) is especially remarkable. This function does not present a singularity in the boundary; however, the $x$-axis has been enlarged so that the localization at the boundary could be better appreciated. The FEM solution is a clear example of an \emph{edge state}. These are eigenfunctions that are associated with negative eigenvalues and that are strongly localized at the boundary of the system, while they vanish in the interior of the manifold (the interior of the interval in this case). These \textrm{edge states} are important in the understanding of certain physical phenomena like the quantum Hall effect (see \cite{Xi92} and references therein). The bad behavior of the eigenfunctions obtained with the CPM routine could be due to the fact that, after an initial estimation of the corresponding eigenvalue, it uses a numerical integrator to propagate the solution from one end of the interval to the other. Then it uses the differences between the obtained boundary conditions and the given ones to provide a new starting point and propagate it back. This procedure continues iteratively until convergence while continuity is imposed in the center of the interval. This last statement could explain why the bad behavior always appears in the middle of the interval. It needs to be said that the results using CPM do not depend on the special form that one selects to implement the boundary conditions. If one uses one of the equivalent forms
\begin{equation}
		\begin{array}{c}
			\Psi'(0)=0\\
			\Psi(2\pi)+\cot(-\theta/2)\cdot \Psi'(2\pi)=0
		\end{array} 
		\quad\text{or}\quad
				\begin{array}{c}
			\Psi'(0)=0\\
			\sin(-\theta/2)\cdot \Psi(2\pi)+\cos(-\theta/2)\cdot\Psi'(2\pi)=0
		\end{array} 
\end{equation}
the results remain the same.
The CPM routine in general worked much faster than the FEM routine in all performed calculations.\footnote{All the numerical calculations of this section were performed with a laptop computer with an Intel Core i5 processor at $2.53$ GHz with 4 GB DDR3 RAM.} However, as stated earlier, the algebraic routine used to solve the generalized eigenvalue problem is not adapted to the structure of the problem and other finite element schemes could be used to improve convergence and efficiency.

\begin{figure}[h]\centering
\includegraphics[height=4.7in]{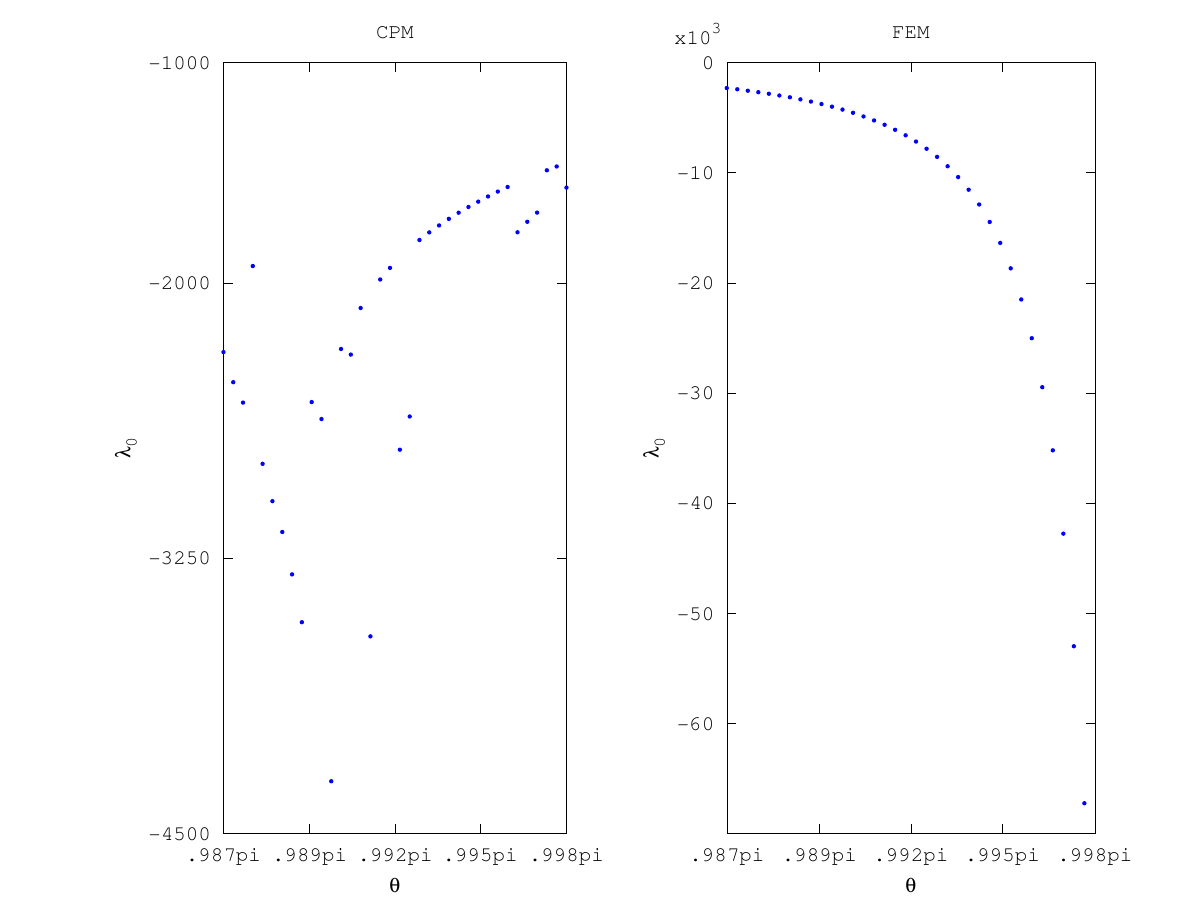}
\caption{Groundlevel energies of the free particle problem for increasing values of $\theta$.}\label{laplaciano2}
\end{figure}

\begin{figure}[h]\centering
\includegraphics[height=3.75in]{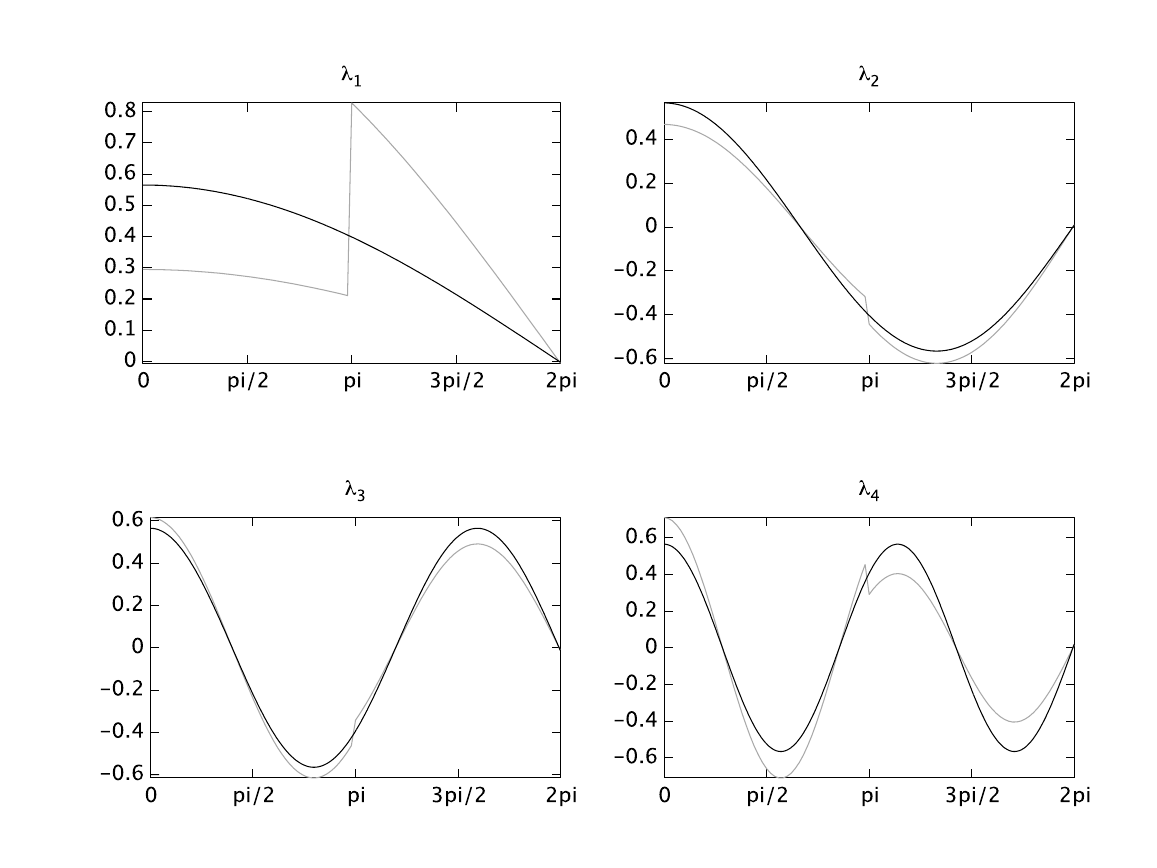}
\caption{Eigenfunctions of the first excited states for the free particle for $\theta=3.1$. FEM functions are plotted in black. CPM functions are plotted in grey.}\label{laplaciano3}
\end{figure}

\begin{figure}[h]\centering
\includegraphics[height=3.75in]{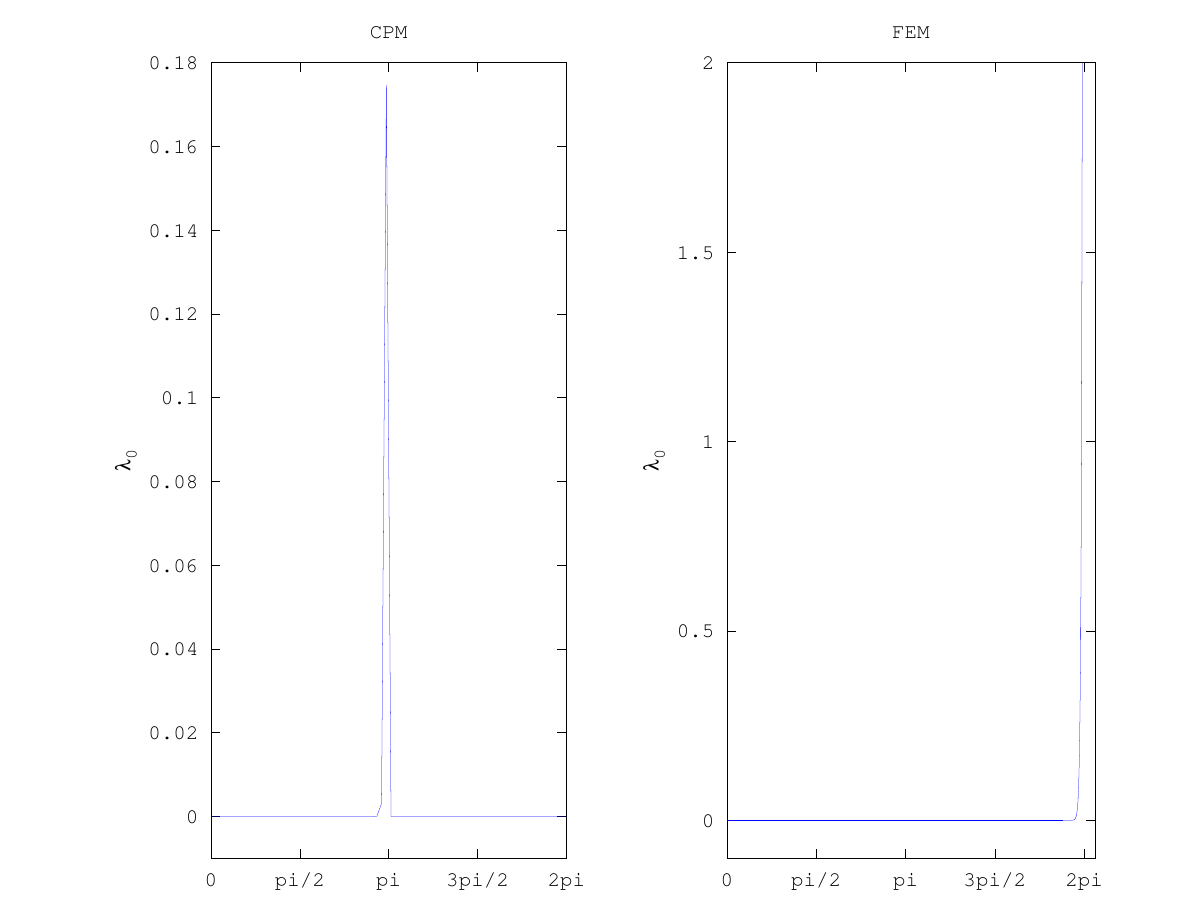}
\caption{Groundlevel eigenfunctions of the free particle problem for $\theta=3.1$.}\label{laplaciano4}
\end{figure}

\FloatBarrier

%%%%%%%%%%%%%%%%%%%%%%%%%%%%%%%%%%%%%%%%%%%%%%%%%%%%%%%%%%%%%%%%%%%%%%%%%%%%%%%%%%%%%%%%%%%%%%%%%%%%%%%%%%%%%%%%%%%%%%%%%%%%

\section{Conclusions}

We can conclude that the algorithm introduced in this article is able to implement any, even nonlocal, self-adjoint extension of the 1D Schr\"odinger problem in a straightforward way. These boundary conditions are more general than the separable ones of the form $\alpha\cdot \Psi(a)+\beta\cdot \Psi'(a)=0$, $\gamma\cdot \Psi(b)+\delta\cdot \Psi'(b)=0$. However, this algorithm, in its present form is not able to solve general Sturm-Liouville problems like the ones adressed in \cite{Le10} if they are not presented in Schr\"odinger form. Due to the convergence and conditioning theorems obtained in sections \ref{conditioningoftheboundarymatrix} and \ref{convergenceofthenumericalscheme} the exact solutions can be achieved even if the original problem is close to being singular as has been shown in section \ref{Numerical_experiments_and_conclusions}.   Moreover, the scheme just discussed is susceptible of being adapted to use known techniques to improve the accuracy of higher order eigenvalues (\cite{An86, VB91}). We leave for subsequent works these tasks as well as a convergence rate analysis or the implementation of other finite element methods like the $p$-method.

The analytical properties of the problem exposed here, i.e. the description of the self-adjoint extensions of the Laplace-Beltrami operator in terms of Eq.\ \eqref{asorey}, as well as the convergence properties both of the finite element method and of the numerical scheme based on it, can be extended to arbitrary dimension and thus going beyond Sturm-Liouville problems.  The main issues to be addressed for that purpose are the regularity properties of the unitary operators appearing in the description of boundary conditions, that in that case are going to be defined on infinite dimensional Hilbert spaces, and the properties of the spectrum of the Laplace-Beltrami operator, that in dimension higher than one may be not discrete for a given self-adjoint extension.   The results presented in this article show that the capability of the finite element method to approximate the domains of linear operators in Hilbert spaces will play a crucial role in developing algorithms that address the aforementioned problem in higher dimensions, thus paving the way to develop the corresponding numerical algorithms.

\end{document}